\documentclass{amsart}

\usepackage{3-sunlet}

\begin{document}

\begin{abstract}
    Phylogenetic networks describe the evolution of a set of taxa for which reticulate events have occurred at some point in their evolutionary history. Of particular interest is when the evolutionary history between a set of just three taxa has a reticulate event. In molecular phylogenetics, substitution models can model the process of evolution at the genetic level, and the case of three taxa with a reticulate event can be modelled using a substitution model on a mixed graph called a 3-sunlet. We investigate a class of substitution models called group-based phylogenetic models on 3-sunlet networks. In particular, we investigate the discrete geometry of the parameter space and how this relates to the dimension of the phylogenetic variety associated to the model. This enables us to give a dimension formula for this variety for general group-based models when the order of the group is odd. 
\end{abstract}

\maketitle

\section{Introduction}

Phylogenetic networks are directed graphs that describe the evolution of a set of taxa for which reticulate events have occurred. Such events, which include, horizontal gene transfer and hybridization, are increasingly being discovered to have occurred between taxa, and the development of methods to reconstruct phylogenetic networks from molecular sequence data is an active area of research. It is therefore important that phylogenetic networks and the models that are placed on them are well understood.

In this work, we focus on phylogenetic network-based substitution models. These are latent-variable Markov models where the state space is a set of biological molecules (usually the four nucleic acids $\{ \rm{A,G,C,T}\}$), and along each edge in the network, a transition matrix gives the probabilities of each possible substitution occurring along that edge (see \cite{gross2018distinguishing}, \cite{Nakhleh2011} for further details). In this work we focus on a family of Markov models called group-based models, so called because the state space of the Markov process is identified with a finite abelian group. In molecular phylogenetics, there are several nucleotide substitution models that are group-based models, such as the Jukes-Cantor (JC) model, the Kimura 2-parameter (K2P) model, and the Kimura 3-parameter (K3P) model. In all these cases, the state space of the four nucleic acids $\{ \rm{A,G,C,T}\}$ is identified with the Klein-four group $\mathbb{Z}/2\mathbb{Z}\times\mathbb{Z}/2\mathbb{Z}$.

For Markov models on phylogenetic networks, the joint probabilities of observing particular patterns at the leaves of the network have  polynomial parameterizations in terms of the numerical parameters of the model, i.e., the substitution rates along each edge and reticulation edge parameters. This makes them amenable to algebraic study, and, in particular, the space of all possible joint probabilities at the leaves is the intersection of an algebraic variety with the probabilty simplex.

In this work we are concerned with the dimension of the variety associated to a phylogenetic network and group-based model.  The dimension of such varieties, in this case, $t$-varieties, is an interesting geometric question in its own right, but also has applications to identifiability and phylogenetic network inference. While \cite{gross2023dimensions} establishes the dimension for most group-based phylogenetic network models, the most elusive has been when the network contains 3-cycles.  In this work, we focus on the smallest phylogenetic network containing a 3-cycle, which is called a \emph{3-sunlet}. While these networks remain the most elusive to understand mathematically, they are perhaps the most important cycles to understand biologically.  Indeed, it is assumed that 3-cycles are the most common cycle motif in true phylogenetic networks, since they indicate hybridization or lateral gene transfer between two very closely related taxa, whereas, larger cycles would indicate such reticulation events between less closely related taxa, which, in many cases, is assumed to be rare. Understanding the dimensions of 3-cycles, can help us establish the statistical property of identifiability, as demonstrated in \Cref{prop:identifiability}.  It can also help us understand how networks with 3-cycles are geometrically embedded within networks with larger cycles, helping interpret residuals when using algebraic methods as in \cite{barton2022statistical}, \cite{martin2023algebraic} or determining the most appropriate penalty term when using Bayesian methods. 

For group-based models on phylogenetic trees, after a transformation,  the parameterization of the model is monomial \cite{evans1993invariants}, \cite{szekely1993fourier}, and thus the corresponding variety is a toric variety. These models have been well studied (see e.g. \cite{sturmfels2005toric}). The parameterization of the model on a sunlet network has a  combinatorial interpretation, and, after the same transformation, is described by binomials. Here we look at the 3-cycle case in depth, and establish a dimension result for group-based models for groups of odd order.

\begin{theorem}\label{thm:main}
Let $G$ be a finite abelian group of odd order $\ell+1 \geq 5$, and let $\mathcal{N}$ be the 3-sunlet network under the general group-based model given by $G$, with corresponding phylogenetic network variety $V_\mathcal{N}^G$. Then the affine dimension of $V_\mathcal{N}^G$ is given by 
$$\dim V_\mathcal{N}^G = 5\ell + 1.$$
\end{theorem}

We call the quantity $5\ell + 1$ the \emph{expected dimension} of the model, and \Cref{thm:main} agrees with the conjecture given in \cite{gross2023dimensions}. As we will see in Section \ref{sec:experiments}, we believe that this conjecture holds for all finite abelian groups (of order at least 5), but as we discuss in detail in Section \ref{sec:discussion}, the proof strategy that we use here is not easily modified for groups of even order, and thus those cases remain open.

The paper is organized as follows.  In Section \ref{sec:background}, we describe phylogenetic networks and the  paramaterization map for 3-sunlet networks for the general group-based model. We also outline the tropicalization method from \cite{DRAISMA2008349} that we use to determine a lower bound for the dimension. The method is rooted in tropical geometry and leads to hyperplane arrangements on spaces of weight vectors. We close Section \ref{sec:background} with observations about the chambers of these hyperplane arrangements for 3-sunlet networks.  In Section \ref{sec:solution}, we prove the main theorem of the paper (Theorem \ref{thm:main}), which gives a formula for the dimension for 3-sunlet networks under general group-based models of odd order greater than or equal to 5.  We end the section with a partial identifiability result (Proposition \ref{prop:identifiability}) for general group-based models of odd order. In Section \ref{sec:experiments}, we investigate the dimension for small
finite abelian groups (both even and odd) through computational experiments. In particular, we explore chambers of hyperplane arrangements to highlight the difficulty in finding appropriate weight vectors that can be used to establish dimensions of 3-sunlets.  Section \ref{sec:discussion} closes the paper with a discussion about the challenges involved in understanding 3-sunlets, and more generally, networks with 3-cycles.

\section{Background}\label{sec:background}

A \emph{(rooted binary) phylogenetic network} is a rooted, acyclic, directed graph where each non-root internal vertex has in-degree one and out-degree two, or in-degree two and out-degree one. We refer to the internal vertices with in-degree one as \emph{tree vertices} and the internal vertices with in-degree two as \emph{reticulation vertices}. The leaves of the phylogenetic network (the vertices of in-degree 1 and out-degree 0) are labelled by a set of taxa, for which we will always use the set of the first $n$ positive integers $[n]=\{1, \ldots, n\}$. The two edges directed into a reticulation vertex are called \emph{reticulation edges}. A phylogenetic network $\mathcal{N}$ is said to be level-1 if, in the undirected skeleton of $\mathcal{N}$, no two cycles share an edge. For an example of a level-1 phylogenetic network, see \Cref{fig:level-1 network}. A \emph{semi-directed phylogenetic network} is a mixed graph that is obtained from a phylogenetic network by suppressing the root vertex and un-directing all non-reticulation edges. Semi-directed networks generalize the notion of unrooted trees, and, for group-based models, if two phylogenetic networks have the same underlying semi-directed topology, then their corresponding varieties are also equal \cite[Lemma 2.2]{gross2023dimensions}. Since we are concerned with the dimensions of the corresponding varieties, we will only consider semi-directed phylogenetic networks.

The fundamental building blocks of level-1 semi-directed phylogenetic networks are unrooted trees and $k$-\emph{sunlet} networks, which are the minimal semi-directed phylogenetic networks containing a $k$-cycle. In this paper, we focus on 3-sunlet networks, which are the minimal semi-directed phylogenetic networks containing a triangle (i.e., a 3-cycle). A 3-sunlet can be obtained from the phylogenetic network in \Cref{fig:level-1 network} by restricting the network to the leaves labelled by taxa 1, 3, and 6 and suppressing vertices of degree 2 (restriction is discussed in more detail towards the end of Section \ref{sec:solution}). Phylogenetic networks with triangles are thought to be among the most common phylogenetic networks, because hybridization usually occurs between closely related species. Despite this, 3-sunlet networks are the least understood of the sunlet networks.

\begin{figure}[h!]
    \centering
    \begin{tikzpicture}
    [every node/.style={inner sep=0pt},
                    every path/.style={thick},   
decoration={markings, 
    mark= at position 0.5 with {\arrow{stealth}}
    }
]
        \draw[postaction={decorate}] (0,1)--(2,-1);
        \draw[postaction={decorate}] (2,-1)--(1,-2);
        \draw[postaction={decorate}] (2,-1)--(3,-2);
        \draw[dashed, postaction={decorate}] (3,-2)--(2,-3);
        \draw[postaction={decorate}] (3,-2)--(4,-3);
        \draw (4,-3)--(4,-4);
        \draw[dashed, postaction={decorate}] (1,-2)--(2,-3);
        \draw[postaction={decorate}] (2,-3)--(2,-4);
        \draw[postaction={decorate}] (1,-2)--(0,-3);
        \draw (0,-3)--(0,-4);
        \draw[postaction={decorate}] (0,1)--(-1,0);
        \draw[dashed, postaction={decorate}] (-1,0)--(-1,-3);
        \draw[postaction={decorate}] (-1,-3)--(-1,-4);
        \draw[postaction={decorate}] (-1,0)--(-2.5,-1.5);
        \draw[dashed, postaction={decorate}] (-2.5,-1.5)--(-1,-3);
        \draw[postaction={decorate}] (-2.5,-1.5)--(-3,-2);
        \draw[postaction={decorate}] (-3,-2)--(-2,-3);
        \draw(-2,-3)--(-2,-4);
        \draw[postaction={decorate}] (-3,-2)--(-4,-3);
        \draw(-4,-3)--(-4,-4);

        \node (v3) at (-4,-4.5) {$1$};
        \node (v3) at (-2,-4.5) {$2$};
        \node (v3) at (-1,-4.5) {$3$};
        \node (v3) at (0,-4.5) {$4$};
        \node (v3) at (2,-4.5) {$5$};
        \node (v3) at (4,-4.5) {$6$};
        
    \end{tikzpicture}
    \caption{A rooted, level-1 phylogenetic network. This network contains a single 3-cycle and a single 4-cycle. Reticulation edges are drawn with dashed lines.}
    \label{fig:level-1 network}
\end{figure}
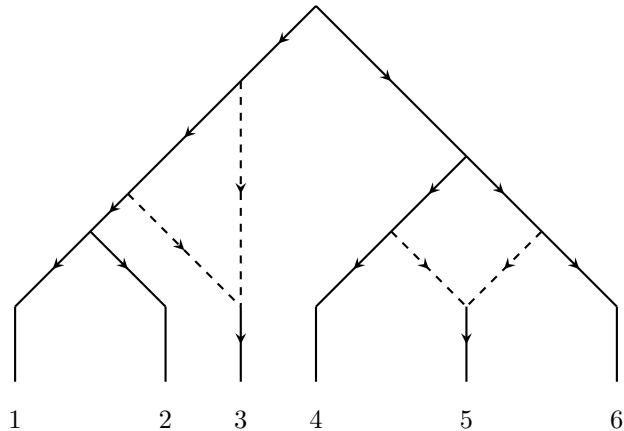

We place a group-based model of evolution on a level-1 semi-directed phylogenetic network $\mathcal{N}$ by arbitrarily assigning direction to all undirected edges (i.e., non-reticulation edges), and choosing a finite abelian group $G$ and a subgroup $B$ of the automorphism group of $G$, denoted $\Aut(G)$.  We note that arbitrarily assigning directions to all undirected edges results in the same variety as rooting the semi-directed network even if the chosen edge directions are not consistent with any placement of a root vertex (see the proof of Lemma 2.2 in \cite{gross2023dimensions}). The group $G$ is identified with the state space of the model, and the group $B$ encodes additional constraints that the transition matrices must adhere to. When we choose $B=\{\text{id}\}$ we call the model the \emph{general group-based model for $G$}. For example, the Kimura 3-parameter (K3P) model is the general group-based model for $G = \mathbb{Z}/2\mathbb{Z}\times\mathbb{Z}/2\mathbb{Z}$. The Jukes-Cantor (JC) model is the group-based model with $G = \mathbb{Z}/2\mathbb{Z}\times\mathbb{Z}/2\mathbb{Z}$ and $B = \Aut{G}$, which we identify with $S_3$, the symmetric group of order 3. In between these two we have the Kimura 2-parameter model (K2P), where $B$ is a subgroup isomorphic to $S_2$. Group-based models have the desirable property that for any phylogenetic tree there exists a Fourier transformation that transforms expressions for the marginal probabilities of observations at the leaves into monomial expressions (see e.g., \cite[Chapter 15]{gsm194sullivant} for an overview). For level-1 phylogenetic networks, that same transformation significantly simplifies the expressions for the marginal probabilities, although they are not monomial \cite[Prop 4.2]{gross2018distinguishing}.

We are interested in identifying the semi-directed phylogenetic network from observed data on the leaves. As noted above, for group-based models, the root of network is not identifiable. For certain group-based models, identifiability results are known (see e.g. \cite{gross2018distinguishing}, \cite{gross2021distinguishing}, \cite{hollering2021identifiability}, \cite{cummings2023pfaffian}), but a general result for all group-based models has yet to be determined. Understanding the dimension of the variety associated to a phylogenetic network and model can assist in determining identifiability. A step in this direction was taken in \cite{gross2023dimensions}, and some identifiability results were obtained for arbitrary group-based models. Here, one limiting factor was being unable to determine the dimension of the varieties associated the the 3-sunlet network.

\subsection{The 3-sunlet networks}
The 3-sunlet is the semi-directed network topology of a simple 3-leaf phylogenetic network with a single cycle. It poses a particular problem to phylogeneticists, because under the most commonly used 4-state group-based models (JC, K2P, and K3P), the reticulation vertex is not identifiable from data at the leaves of the network (see e.g. \cite[Lemma~1]{gross2021distinguishing}). Thus many of the identifiability results obtained for these models require the phylogenetic networks to be \lq triangle free\rq ~\cite{gross2021distinguishing}.

Mathematically, the 3-sunlet is a mixed graph consisting of a single (undirected) 3-cycle, and one leaf vertex adjacent to each vertex in the cycle. One of the vertices in the cycle is the reticulation vertex, and the two cycle edges adjacent to this vertex are reticulation edges. The reticulation edges are the only directed edges and they are directed towards the reticulation vertex (see Figure \ref{fig:3sunlet} for an example). By removing either of the edges $e_6$ or $e_5$ in \Cref{fig:3sunlet} and undirecting the remaining edge, we obtain an unrooted phylogenetic tree (with a vertex of degree 2), which we denote by $\C{T}_1$ and $\C{T_2}$ respectively.

 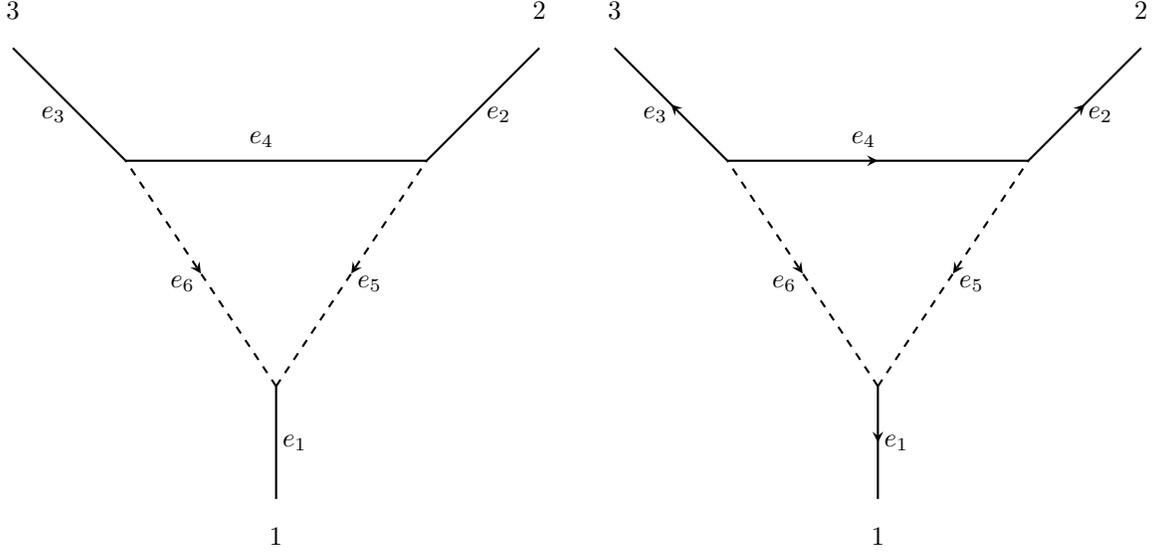
\begin{figure}[ht]
 \centering
\begin{tikzpicture}
[every node/.style={inner sep=0pt},
                    every path/.style={thick},   
decoration={markings, 
    mark= at position 0.5 with {\arrowreversed{stealth}}
    }
]

\draw (0.5,3.5) -- (2,2) node[midway,below=4, left=1] {$e_{3}$};
\draw (6,2)  -- (2 ,2)  node[midway,above=8,left=0.5] {$e_{4}$};
\draw (7.5,3.5) -- (6,2) node[midway,below=4, right=1] {$e_{2}$};
\draw[dashed,postaction={decorate}] (4,-1)  -- (2,2) node[midway,below=4,left=2] {$e_6$};
\draw[dashed,postaction={decorate}] (4,-1)  -- (6,2) node[midway,below=4,right=2] {$e_5$};
\draw (4,-2.5) -- (4,-1) node[midway,right=2] {$e_1$};

\node (v3) at (4,-3) {$1$};
\node (v3) at (7.5,4) {$2$};
\node (v3) at (0.5,4) {$3$};

\draw[postaction={decorate}] (8.5,3.5) -- (10,2) node[midway,below=4, left=1] {$e_{3}$};
\draw[postaction={decorate}]  (14,2)  -- (10 ,2)  node[midway,above=8,left=0.5] {$e_{4}$};
\draw[postaction={decorate}]  (15.5,3.5) -- (14,2) node[midway,below=4, right=1] {$e_{2}$};
\draw[dashed,postaction={decorate}] (12,-1)  -- (10,2) node[midway,below=4,left=2] {$e_6$};
\draw[dashed,postaction={decorate}] (12,-1)  -- (14,2) node[midway,below=4,right=2] {$e_5$};
\draw[postaction={decorate}]  (12,-2.5) -- (12,-1) node[midway,right=2] {$e_1$};

\node (v3) at (12,-3) {$1$};
\node (v3) at (15.5,4) {$2$};
\node (v3) at (8.5,4) {$3$};

\end{tikzpicture}
\caption{(Left) The semi-directed network topology of the 3-sunlet network with taxa labels 1,2, and 3. (Right) A directed 3-sunlet network. }
\label{fig:3sunlet}
\end{figure}

In order to simplify our exposition, we begin with the phylogenetic network parameterization in the transformed coordinates. Readers interested in the derivation of this parameterization from the substitution model can consult \cite[Section~2]{gross2023dimensions} for a full explanation.  In order to specify the parameterization we must direct the undirected edges of the semi-directed network topology. By \cite[Lemma~2.2]{gross2023dimensions} we may arbitraily choose these directions, and so for the remainder of this work we denote by $\C N$ the directed 3-sunlet network in \Cref{fig:3sunlet} (right).

Let $G$ be a finite abelian group and let $B$ be a subgroup of $\Aut{G}$. Let $B\cdot G$ be the set of $B$-orbits of $G$, and define $\ell + 1 := |B\cdot G|$. Note that when $B=\{\rm{id}\}$ we have $\ell + 1 = |G|$. A \emph{consistent leaf-labelling} is a triple $\mathbf{g} = (g_1, g_2, g_3)\in G^3$ satisfying $g_1 + g_2 + g_3 = 0$. For a fixed $G$ there are exactly $|G|^2$ consistent leaf-labellings. We give $\mathbb{C}^{6(\ell + 1)}$ a basis indexed by $B$-orbits and edges of $\mathcal{N}$, and denote the basis element corresponding to the $B$-orbit $[g]$ and edge $e_i$ as $E_i^g$. We give $\mathbb{C}^{(\ell + 1)^2}$ a basis indexed by consistent leaf-labellings $\mathbf{g} = (g_1, g_2, g_3)$. Then the \emph{parameterization map} $\phi_{\mathcal{N}}^{(G,B)}$ (in transformed coordinates) is given by
\begin{align*}
&\phi_{\mathcal{N}}^{(G,B)}:\, \mathbb{C}^{6(\ell + 1)} \to \mathbb{C}^{(\ell + 1)^2}, \\
                             \big(\phi_{\mathcal{N}}^{(G,B)}(w)\big)_{\mathbf{g}} &= w_1^{g_1}w_2^{g_2}w_3^{g_3}w_4^{g_1 + g_2}w_5^{g_1} + w_1^{g_1}w_2^{g_2}w_3^{g_3}w_4^{g_2}w_6^{g_1} \\
                             &=m_1(\mathbf{g}) + m_2(\mathbf{g}).  
\end{align*}
where $w_i^g$ is the coefficient of $E_i^g$ in $w$. Here, the first term $m_1(\mathbf{g}) := w_1^{g_1}w_2^{g_2}w_3^{g_3}w_4^{g_1 + g_2}w_5^{g_1}$ comes from the phylogenetic tree $\C{T}_1$ (obtained from $\C{N}$ be removing the edge $e_6$). Each superscript $g_i$ is given by the edge-labelling of the edge $e_i$ in the right hand diagram in \Cref{fig:3sunlet} (see \cite{gross2023dimensions} for further details). Similarly, the second term $m_2(\mathbf{g}) := w_1^{g_1}w_2^{g_2}w_3^{g_3}w_4^{g_2}w_6^{g_1}$ comes from the phylogenetic tree $\C{T}_2$ (obtained from $\C{N}$ by removing the edge $e_5$). The \emph{phylogenetic variety of $\mathcal{N}$ and $(G,B)$} is defined as the Zariski closure of the image of $\phi_{\mathcal{N}}^{(G,B)}$, denoted
$$V_{\mathcal{N}}^{(G,B)} = \overline{\rm{im}\,\phi_{\mathcal{N}}^{(G,B)}},$$
and this is the object that we study. Since the map $\phi_{\mathcal{N}}^{(G,B)}$ is homogeneous, the variety $V_{\mathcal{N}}^{(G,B)}$ is a projective variety. However, we will mostly remain in affine space and consider the affine cone.

Observe that the map $\phi_{\mathcal{N}}^{(G,B)}$ is a morphism of affine varieties. It has comorphism given by
\begin{align*}
\psi_{\mathcal{N}}^{(G,B)}:\, \mathbb{C}[q_{\mathbf{g}}\ &|\ g_1 + g_2 + g_3 = 0] \to \mathbb{C}[a_i^g\ |\ i = 1,\ldots,6,\text{ and } g\in B\cdot G], \\
    q_{\mathbf{g}} &\longmapsto  a_1^{g_1}a_2^{g_2}a_3^{g_3}a_4^{g_1 + g_2}a_5^{g_1} + a_1^{g_1}a_2^{g_2}a_3^{g_3}a_4^{g_2}a_6^{g_1}, 
\end{align*}
where we think of the coordinate ring of $\mathbb{C}^{(\ell + 1)^2}$ as being generated by variables $q_{\mathbf{g}}$ for consistent leaf-labellings $\mathbf{g} = (g_1,g_2,g_3)$, and we think of the coordinate ring of $\mathbb{C}^{6(\ell + 1)}$ as being generated by variables $a_i^g$ for $i=1,\ldots,6$ and $g \in B\cdot G$. Under this definition, the zero ideal of $V_{\mathcal{N}}^{(G,B)}$, which we denote $I_{\mathcal{N}}^{(G,B)}$, is given by $\ker \psi_{\mathcal{N}}^{(G,B)}$. This ideal, and in particular, its generating sets, are an important object of study in mathematical phylogenetics as they can be used for model selection \cite{barton2022statistical} \cite{cummings2021invariants} \cite{martin2023algebraic} and to establish identifiability \cite{gross2018distinguishing} \cite{hollering2021identifiability} \cite{gross2021distinguishing} \cite{cummings2023pfaffian}. While some polynomials are known for some group-based models such as CFN \cite{cummings2021invariants} \cite{cummings2023pfaffian} and several 4-state models \cite{frohn2024invariants}, generating sets of $I_{\mathcal{N}}^{(G,B)}$ are not known in general.

\subsection{Determining Dimension}
In Section \ref{sec:solution}, we give a dimension result for the 3-sunlet network and the general group-based model for groups of odd order by following the approach taken in \cite{gross2023dimensions}. Here, we introduce the concepts and objects we need. First, for a level-1 phylogenetic network $\mathcal{N}$ and group-based model $(G,B)$, the variety $V_{\C N}^{(G,B)}$ is defined as the Zariski closure of the image of a polynomial map (as defined above for the 3-sunlet). Considering the number of free parameters in the domain of the map gives us an upper bound on the dimension.
\begin{lem}\cite[Proposition~4.2]{gross2023dimensions}\label{lem:upper}
    Let $\mathcal{N}$ be the 3-sunlet network, $G$ a finite abelian group and $B$ a subgroup of $\rm{Aut}(G)$ with $|G\cdot B| = \ell + 1$. Then
    $$\dim V_{\C N}^{(G,B)} \leq 5\ell + 1.$$
    \qed
\end{lem}

As a result of Lemma \ref{lem:upper}, to prove Theorem \ref{thm:main}, it is sufficient for us to give a lower bound on the dimension. We do this by exhibiting a Jacobian matrix of sufficient rank of the tropicalization of the parameterization.

Let $\phi = \phi_{\mathcal{N}}^{(G,B)}: \mathbb{C}^{6(\ell + 1)} \to \mathbb{C}^{(\ell+1)^2}$ be the parameterization map of the 3-sunlet network under the general group-based model for an abelian group $G$, and for a consistent leaf-labelling $\mathbf{g} = (g_1, g_2, g_3)$, let $\phi_{\mathbf{g}}$ be the component of $\phi$ mapping onto the $\mathbf{g}$-coordinate of $\mathbb{C}^{(\ell+1)^2}$. Since $\phi$ is a polynomial map, each $\phi_{\mathbf{g}}$ is a polynomial, and we can define
\begin{align*}
\Trop(\phi_{\mathbf{g}}):\,&\mathbb{R}^{6(\ell + 1)} \to\mathbb{R} \\
                        &\lambda  \mapsto \min_{\alpha\in M} \langle\lambda,\alpha\rangle,
\end{align*}
where $M\subset \mathbb{Z}^{6(\ell + 1)}$ denotes the set of exponent vectors corresponding to the monomials in the polynomial expression for $\phi_{\mathbf{g}}$. Then $\Trop(\phi): \mathbb{R}^{6(\ell + 1)} \to\mathbb{R}^{(\ell + 1)^2}$ is the map with components $\Trop(\phi_{\mathbf{g}})$. Now for $\lambda\in\mathbb{R}^{6(\ell + 1)}$ at which $\Trop(\phi)$ is differentiable there exists a matrix $A_\lambda$ such that $\Trop(\phi)(\mu) = A^T_\lambda \mu$ for all $\mu$ in an open neighbourhood of $\lambda$ (in fact $A^T_\lambda$ is the Jacobian of $\Trop(\phi)$ at $\lambda$). Then we have a lower bound on the dimension of the affine variety $V_\mathcal{N}^G$ given by
$$ \dim V_{\mathcal{N}}^G \geq \max_{\lambda\in \mathbb{R}^{6(\ell + 1)}} \rank_\mathbb{R} A_\lambda.$$
This is a specific case of Corollary 2.3 in \cite{DRAISMA2008349}. For full details we recommend the reader consult \cite{DRAISMA2008349} and \cite[Section~2.3]{gross2023dimensions}.

In this paper, we study the matrices $A_\lambda$, and, in particular, the cone in the space $\mathbb{R}^{6(\ell + 1)}$ that they induce. We can think about this space as being a kind of \lq tropical dual\rq\ to the parameter space $\mathbb{C}^{6(\ell + 1)}$, and we adopt the same indexing. That is, the entries of $\lambda\in\mathbb{R}^{6(\ell + 1)}$ are indexed by $B$-orbits and edges of $\C{N}$. Then $\lambda(w_i^g) = \lambda_i^g $, where $i\in\{1,2,3,4,5,6\}$ and $g\in B\cdot G$. Observe that the vector $\lambda$ defines a monomial order on the polynomial ring $\mathbb{C}[a_i^g\ |\ i = 1,\ldots,6,\text{ and } g\in B\cdot G]$ (provided we specify another order for resolving ties). For this reason, we call $\lambda\in\mathbb{R}^{6(\ell + 1)}$ a \emph{weight vector}.

For weight vectors $\lambda$ where $\Trop(\phi)$ is differentiable, each column of $A_\lambda$ is indexed by a consistent leaf-labelling $\mathbf{g}=(g_1, g_2, g_3)$. The $\mathbf{g}$-entry of the parameterization $\phi_{\C{N}}$ is given by $m_1(\mathbf{g}) + m_2(\mathbf{g})$, where $m_i(\mathbf{g})$ is the $\mathbf{g}$-entry of the parameterization of $\C{T}_i$. For each $\mathbf{g}$ and $m_i = m_i(\mathbf{g})$ we define $\lambda(m_i)$ to be the natural product of the row vector of exponents of $m_i$, with the column vector $\lambda$. That is,

\begin{gather}\label{eqn:lambdam1}
    \begin{aligned}
    \lambda(m_1) &= \lambda(w_1^{g_1}w_2^{g_2}w_3^{g_3}w_4^{g_1+g_2}w_5^{g_1}) \\
    &= \lambda_1^{g_1} + \lambda_2^{g_2} + \lambda_3^{g_3} + \lambda_4^{g_1 + g_2} + \lambda_5^{g_1}
    \end{aligned}
\end{gather}
and
\begin{gather}\label{eqn:lambdam2}
\begin{aligned}
    \lambda(m_2) &= \lambda(w_1^{g_1}w_2^{g_2}w_3^{g_3}w_4^{g_2}w_6^{g_1}) \\
    &= \lambda_1^{g_1} + \lambda_2^{g_2} + \lambda_3^{g_3} + \lambda_4^{g_2} + \lambda_6^{g_1}.
\end{aligned}
\end{gather}
Thus for a given weight vector $\lambda$, we take the column of $A_\lambda$ indexed by $\mathbf{g}$ to be the exponent vector of the monomial $m_i$ where $\lambda(m_i) < \lambda(m_j)$ for $i\neq j \in \{1,2\}$; in this case, we will say that $\lambda$ \emph{assigns} $\mathbf{g}$ to tree $\mathcal T_i$. The procedure just described is equivalent to describing an initial ideal of the ideal generated by the image of $\psi_{\mathcal{N}}^{(G,B)}$, where the initial term for each generator, $\text{in}_\lambda(\psi_{\mathcal{N}}^{(G,B)}(q_{\mathbf{g}}))$, is given by the monomial $m_i$.

In Section \ref{sec:solution}, we give a solution to the dimension problem for general group-based models on the 3-sunlet network for all finite abelian groups of odd order at least 7, by constructing $\lambda$ such that $A_\lambda$ has maximal rank.

\subsection{Defining Hyperplanes}

\label{subsec: hyperplanes}

The matrix $A_\lambda$ is determined by inequalities between linear combinations of the coordinates of the weight vector $\lambda$. In this subsection, we construct a hyperplane arrangement $\mathcal{H}_G$, which divides Euclidean space into the regions on which $\Trop(\phi)$ is differentiable and $A_\lambda$ is constant. The hyperplanes themselves correspond to regions on which $\Trop(\phi)$ is not differentiable. By understanding the defining hyperplanes and resulting geometry, we are able to construct a weight vector $\lambda$ so that $A_\lambda$ has the maximum possible rank, and therefore gives the best lower bound on the dimension of the variety $V_{\C{N}}^{(G,B)}$.

\begin{defn}
    \label{def:hyperplane-arrangement}
    A \textit{hyperplane arrangement} is a collection of hyperplanes $\C{H} = \{ H_i \}_{i \in \C{I}}$, $H_i \subset \RR^n$. A connected component of the complement $\RR^n \setminus \C{H}$ is a \textit{chamber of } $\C{H}$. We denote by $C(\C{H})$ the set of chambers of $\C{H}$.
\end{defn}

For the 3-sunlet network, each column of $A_\lambda$ corresponds to a consistent leaf-labelling $\mathbf{g} = (g_1, g_2, g_3)$. The column of $A_\lambda$ corresponding to $\mathbf{g}$ can be one of two vectors, and depends on which inequality the coordinates of $\lambda$ satisfy. Thus, each consistent leaf-labelling $\mathbf{g}$ determines a hyperplane of the arrangement we want to construct for the 3-sunlet network.
As above, for $i=1,2$ we write $m_i(\mathbf{g})$ for the monomial corresponding to the $i^{\text{th}}$ tree in the $\mathbf{g}$-entry of the parametrization map $\phi_\C{N}$. Then, we can define the  hyperplane arrangement corresponding to the 3-sunlet as the collection $\C{H} = \{H_{\mathbf{g}}\,|\ \mathbf{g}=(g_1,g_2,g_3)$ is a consistent leaf-labelling$\}$, where
\begin{align*}
\centering
\begin{split}
    H_\mathbf{g} :&= \{ \lambda \in \RR^{6(\ell+1)} \mid \lambda(m_1(\mathbf{g})) = \lambda(m_2(\mathbf{g})) \}\\ 
    &= \{ \lambda \in \RR^{6(\ell+1)} \mid \lambda_1^{g_1} + \lambda_2^{g_2} + \lambda_3^{g_3} + \lambda_4^{g_1 + g_2} + \lambda_5^{g_1} = \lambda_1^{g_1} + \lambda_2^{g_2} + \lambda_3^{g_3} + \lambda_4^{g_2} + \lambda_6^{g_1} \}.
\end{split}
\end{align*}

In \Cref{lem:inequalities}, we make several observations about the hyperplanes $H_{\mathbf{g}}$. Note that the hyperplane arrangement $\C{H}$ and the resulting chambers $C(\C{H})$ form the Gr\"{o}bner fan of the ideal generated by the image of $\psi_{\mathcal{N}}^{(G,B)}$, denoted $\langle\text{im}\,\psi_{\mathcal{N}}^{(G,B)}\rangle$. In this interpretation, the hyperplanes themselves consist of those points $\lambda$ for which the initial ideal $\text{in}_\lambda(\langle\text{im}\,\psi_{\mathcal{N}}^{(G,B)}\rangle)$ is not a monomial ideal, and the chambers consist of those points $\lambda$ for which the initial ideal $\text{in}_\lambda(\langle\text{im}\,\psi_{\mathcal{N}}^{(G,B)}\rangle)$ is monomial.

We make the following observation about the matrix $A_\lambda$.

\begin{lem}\label{lem:inequalities}
    Let $\mu_g := \lambda_6^g - \lambda_5^g$ for all $g \in G$. The inequalities determining the matrix $A_\lambda$ are:
    \begin{align}
        0 & < \mu_0\  \text{ or \ } \mu_0<0, \emph{ and } \label{eqn:hyperplane1} \\
        \lambda_4^{g_1 + g_2} - \lambda_4^{g_2} & < \mu_{g_1} \text{ or \ }   \mu_{g_1} <  \lambda_4^{g_1 + g_2} - \lambda_4^{g_2}\text{ for all } g_1, g_2 \in G \text{ such that } g_1 \neq 0.\label{eqn:hyperplane2}
    \end{align}
    In particular, the determining inequalities depend only on $g_1$ and $g_2$.
\end{lem}

\begin{proof}
    The matrix $A_\lambda$ is constant on a region $R$ if and only if each column is constant on $R$, so it follows that the hyperplane arrangement we seek is the union of hyperplane arrangements each of whose regions defines the constant regions of a single column.

    The assignment of columns is equivalent to choosing the direction of the inequality in $\lambda(m_1) < \lambda(m_2)$ or $\lambda(m_1) > \lambda(m_2)$. Without loss of generality, assume $\lambda(m_1) < \lambda(m_2)$. Expanding $\lambda(m_1)$ and $\lambda(m_2)$ as in equations (\ref{eqn:lambdam1})
and (\ref{eqn:lambdam2}) we obtain
$$ \lambda_1^{g_1} + \lambda_2^{g_2} + \lambda_3^{g_3} + \lambda_4^{g_1 + g_2} + \lambda_5^{g_1} < \lambda_1^{g_1} + \lambda_2^{g_2} + \lambda_3^{g_3} + \lambda_4^{g_2} + \lambda_6^{g_1} $$
    Cancelling terms we see that
    \begin{align*}
        \lambda(m_1) < \lambda(m_2) \iff & \lambda_4^{g_1 + g_2} + \lambda_5^{g_1} 
        <
        \lambda_4^{g_2} + \lambda_6^{g_1} \\
        \iff & \lambda_4^{g_1 + g_2} - \lambda_4^{g_2} < \lambda_6^{g_1} - \lambda_5^{g_1}
        \\
        \iff & \lambda_4^{g_1 + g_2} - \lambda_4^{g_2} < \mu_{g_1}.
    \end{align*}
\end{proof}
Inequality (\ref{eqn:hyperplane1}) controls the assignment of the $|G| = \ell+1$ columns with label $(0, g_2, -g_2)$, for $g_2 \in G$. The inequalities given by (\ref{eqn:hyperplane2}) each control a single column with label $(g_1, g_2, -g_1 - g_2)$, with $g_1 \neq 0$. Thus, by crossing the hyperplane $\mu_0=0$ exactly $\ell + 1$ of the columns of $A_\lambda$ change. While, when crossing a hyperplane of the form $ \mu_{g_1} = \lambda_4^{g_1 + g_2} - \lambda_4^{g_2} $, exactly one of the columns of $A_\lambda$ changes. Observe that we cannot simply choose which inequalities are satisfied and expect to find a weight vector $\lambda$ that achieves this. That is, some combinations of inequalities cannot be simultaneously satisfied.  Thus, for a given $G$, it is unclear how many chambers lie in the hyperplane arrangement $\mathcal{H}$, but it is at most $2^{\ell(\ell +1) + 1}$. \Cref{lem:ginverseg} in the next section gives some restrictions on these assignments for groups of odd order, and we will see further examples in \Cref{sec:experiments}.

\begin{example}
    Let $G = \mathbb{Z}/2\mathbb{Z}$. We have four consistent leaf-labellings given by $(0,0,0), (0,1,1), (1,0,1), (1,1,0)$. The image of $\phi_{\mathcal{N}}^{G}$ is given by
    \begin{align*}
        \phi_{\mathcal{N}}^{G}(w)_{000} &= w_1^{0}w_2^{0}w_3^{0}w_4^{0}w_5^{0} + w_1^{0}w_2^{0}w_3^{0}w_4^{0}w_6^{0}, \\
        \phi_{\mathcal{N}}^{G}(w)_{011} &= w_1^{0}w_2^{1}w_3^{1}w_4^{1}w_5^{0} + w_1^{0}w_2^{1}w_3^{1}w_4^{1}w_6^{0}, \\
        \phi_{\mathcal{N}}^{G}(w)_{101} &= w_1^{1}w_2^{0}w_3^{1}w_4^{1}w_5^{1} + w_1^{1}w_2^{0}w_3^{1}w_4^{0}w_6^{1}, \\
        \phi_{\mathcal{N}}^{G}(w)_{110} &= w_1^{1}w_2^{1}w_3^{0}w_4^{0}w_5^{1} + w_1^{1}w_2^{1}w_3^{0}w_4^{1}w_6^{1}, \\
    \end{align*}
    where in each case the first monomial, $m_1$ corresponds to $\mathcal{T}_1$, and the second monomial $m_2$ corresponds to $\mathcal{T}_2$. Pick a weight vector $\lambda$ with $\lambda_4^0 = 3, \lambda_4^1 = 1, \mu_0 = 1$, and $\mu_1 = -1$. We construct the matrix $A_\lambda$. The columns of $A_\lambda$ are indexed by the 4 consistent leaf-labellings, and the rows are indexed by the 12 parameters $w_i^g$ for $i=1,\ldots,6$ and $g = 0,1$. For columns $(0,0,0)$ and $(0,1,1)$ we have that $0 < \mu_0 = 1$, so in these cases the monomial from $\mathcal{T}_1$ is chosen. For $(1,0,1)$ we have that $\lambda_4^1 - \lambda_4^0 = -2 < \mu_1 = -1$ so the monomial from $\mathcal{T}_1$ is chosen. For $(1,1,0)$ we have that $\lambda_4^0 - \lambda_4^1 = 2 > \mu_1 = -1$ so the monomial from $\mathcal{T}_2$ is chosen. This gives us the following matrix $A_\lambda$
    \[ A_\lambda = 
\begin{blockarray}{ccccc}
& \mathcal{T}_1 & \mathcal{T}_1 & \mathcal{T}_1 & \mathcal{T}_2 \\ [2pt]
& 000 & 011 & 101 & 110 \\ [2pt]
\begin{block}{c(cccc)}
  w_1^0 & 1 & 1 & 0 & 0 \\ [3pt]
  w_1^1 & 0 & 0 & 1 & 1 \\ [3pt]
  w_2^0 & 1 & 0 & 1 & 0 \\ [3pt]
  w_2^1 & 0 & 1 & 0 & 1 \\ [3pt]
  w_3^0 & 1 & 0 & 0 & 1 \\ [3pt]
  w_3^1 & 0 & 1 & 1 & 0 \\ [3pt]
  w_4^0 & 1 & 0 & 0 & 1 \\ [3pt]
  w_4^1 & 0 & 1 & 1 & 0 \\ [3pt]
  w_5^0 & 1 & 1 & 0 & 0 \\ [3pt]
  w_5^1 & 0 & 0 & 1 & 0 \\ [3pt]
  w_6^0 & 0 & 0 & 0 & 0 \\ [3pt]
  w_6^1 & 0 & 0 & 0 & 1 \\ [3pt]
\end{block}
\end{blockarray}\ ,
 \]
 where the entry corresponding to column $g_1g_2g_3$ and row $w_i^g$ is the exponent of $w_i^g$ in the monomial from the $g_1g_2g_3$ entry of the expression for $\phi_\mathcal{N}^G(w)$ corresponding to the tree chosen.
\end{example}

To end this section, we make some observations on the hyperplane arrangement $\C{H}$ coming from the 3-sunlet and a group-based model $(G,B)$. We will define the $\emph{rank}$ of a chamber $C \in C(\C{H})$ as the rank of the corresponding tropical Jacobian matrix $A_\lambda$ for all $\lambda\in C$.

First, observe that when we choose $\lambda$ so that all columns correspond to $\mathcal{T}_1$ or $\mathcal{T}_2$ (this is possible by choosing $\mu_g$ either very large or very small for all $g$), then $A_\lambda$ is equal to the corresponding tropical Jacobian for the group-based phylogenetic tree model for $\mathcal{T}_1$ or $\mathcal{T}_2$ (albeit with some extra rows of 0's corresponding to the parameters $w_6^g$ or $w_5^g$ respectively). Since, the matrix $A_\lambda$ has rank equal to the dimension of the toric variety corresponding to $A_{\lambda}$ (see e.g., \cite[Lemma~4.2]{sturmfels1996grobner}), and these varieties are well studied for trees, in both cases, $\rank A_\lambda = 3\ell + 1$ (see e.g. \cite[Lemma~4.1]{gross2023dimensions}). This describes two chambers with rank equal to $3\ell + 1$. In fact, there are two more, as the next proposition demonstrates.

\begin{proposition}\label{prop:hyperplane0rank}
    Let $C$ and $C'$ be two adjacent chambers in $C(\C{H})$, separated by the hyperplane $\mu_0=0$. Then $\rank C = \rank C'$.
\end{proposition}
\begin{proof}
    Let $\lambda \in C$ and $\lambda' \in C'$ be weight vectors, and consider the matrices $A_\lambda$ and $A_{\lambda'}$. Suppose, without loss of generality, that in $A_\lambda$ the columns indexed by $(0,g,-g)$ with $g\in G$ correspond to $\C T_1$, and in $A_{\lambda'}$ they correspond to $\C T_2$. Thus, in the columns of $A_{\lambda}$ indexed by $(0,g,-g)$, entries in the row indexed by $w_5^0$ are 1, and entries in the row indexed by $w_6^0$ are 0. On the other hand, in the columns of $A_{\lambda'}$ indexed by $(0,g,-g)$, entries in the row indexed by $w_5^0$ are 0, and entries in the row indexed by $w_6^0$ are 1. All other entries of $A_{\lambda}$ and $A_{\lambda'}$ are the same. Finally, observe that in the rows indexed by $w_5^0$ and $w_6^0$, all entries are $0$ outside the columns indexed by $(0,g,-g)$ for $g\in G$. Thus, the difference between $A_{\lambda}$ and $A_{\lambda'}$ is a swap of the rows indexed by $w_5^0$ and $w_6^0$, and this does not affect the rank.
\end{proof}

Applying \Cref{prop:hyperplane0rank} to the two chambers of rank $3\ell + 1$ found above, we have four chambers of this rank. These come from column assignments where all columns corresponding to a leaf-labelling $(0,g,-g)$ are assigned to $\C{T}_i$, and all other columns assigned to $\C{T}_j$ for $i,j\in \{1,2\}$. It is easy to see that all other chambers have rank strictly greater than this (changing any column corresponding to $(g_1, g_2, g_3)$ with $g_1\neq 0$ will introduce a $1$ into the row $w_5^{g_1}$ or $w_6^{g_1}$ which was previously all 0's), thus we have exactly four chambers of minimal rank $3\ell + 1$.

\begin{lem}\label{lem:minimalRankChambers}
    Let $\C{H}$ be the hyperplane arrangement given by the 3-sunlet network and group-based model $(G,B)$. Then $\C{H}$ has exactly four chambers of minimal rank $3\ell + 1$.\qed
\end{lem}

\begin{proposition}
    Let $\C{H}$ be the hyperplane arrangement given by the 3-sunlet network and group-based model $(G,B)$. Given a chamber $C\in C(\C{H})$ of rank $r$, the rank of each adjacent chamber is between $r-1$ and $r+1$.
\end{proposition}
\begin{proof}
As described above, moving to an adjacent chamber is equivalent to either swapping the assignment of all columns indexed by $(0,g,-g)$ for $g\in G$, or swapping the assignment of a single column indexed by $(g,h,-g,-h)$ with $g\neq 0$. In the first case, the rank does not change by Proposition \ref{prop:hyperplane0rank}. In the second case, the rank can change by at most $1$.
\end{proof}

\Cref{prop:hyperplane0rank} above shows that every chamber is adjacent to at least one other chamber of equal rank. We would also like to know whether every chamber is adjacent to another of chamber of strictly greater rank, or equivalently, whether the only maximal chambers (with respect to rank) are the globally maximal chambers. As we will see in \Cref{subsec:maximalChambers}, the answer to this is no, and there do exist locally maximal chambers.

\section{Dimension for Odd Order Groups}
\label{sec:solution}
In this section we give the dimension of the 3-sunlet variety for general group-based models where $G$ is a finite abelian group of odd order at least $5$. Our results agree with \cite[Conjecture~7.1]{gross2023dimensions}. Our method is to find a weight vector $\lambda$ so that the corresponding tropical Jacobian $A_\lambda$ has rank greater than or equal to $5\ell + 1$, where $|G| = \ell + 1$, as detailed in Section \ref{sec:background}.

First we set out some notation. Assume $|G|$ is odd and let $|G| = 2t + 1$. Choose a subset $X \subset G$ with $|X| = t$ and satisfying the property that if $g\in X$ then $-g \not\in X$. In particular, $0 \not\in X$. For example, for $G$ equal to the cyclic group of order $2t + 1$ we could have $X = \{ 1,2,\ldots, t\}$. We will need the following lemma.

\begin{lem}\label{lem:hg}
    Let $G$ be an abelian group with $|G| = 2t + 1$, where $t\in\mathbb{N}$ and $t \geq 3$, i.e. $|G|$ is odd with $|G| \geq 7$. Let $X\subset G$ be a subset with $|X| = t$ such that if $g\in X$ then $-g \not\in X$. Then there exists an injective function $h:X \to G\setminus X$ such that $h_g:=h(g)$ is not equal to $ 0,-g,$ or $2g$.
\end{lem}
\begin{proof}
    Let $X = \{g_1, \ldots, g_t\}$ and let $k_i = -g_i \not\in X$. We will give an iterative method of choosing the value of $h_{g_i}$, where at each step we update the set of available choices. 
    To begin with, let $K = \{k_1, \ldots, k_t\}$, this will be our initial set of available choices. For each $i = 1,\ldots, t-3$, choose the value of $h_{g_i}$ to be an element of $K$ satisfying $h_{g_i} \neq k_i$ and $h_{g_i} \neq 2g_i$ (should $2g_i$ be in $K$), and remove the value of $h_{g_i}$ from $K$. This is possible, since at each stage, $K$ contains at least $3$ elements.
    
    Now it remains to choose $h_{g_{t-2}}, h_{g_{t-1}}$ and $h_{g_{t}}$. In the worst case, we have $K = \{k_{t-2}, k_{t-1}, k_t\} = \{2g_{t-2}, 2g_{t-1}, 2g_{t}\}.$ For each $g_i$ we take $h_{g_i}$ to be the (possibly unique) $k_j$ such that $i\neq j$ and $k_j \neq 2g_i$. 
\end{proof}

The next lemma demonstrates that there are relationships among the inequalities in \Cref{lem:inequalities}. 

\begin{lem}\label{lem:ginverseg}
Fix $g\in G$ such that $g \neq -g$, and suppose that we have a weight vector $\lambda$ such that for some $h\in G$ the consistent leaf labelling $(g,h,-g-h)$ is assigned to $\C T_1$ and $(g, h', -g-h')$ is assigned to $\C T_2$ for all $h'\ \neq h$. If there exists $k \in G$ such that $(-g, k, g - k)$ is assigned to $\C T_2$ then $(-g, g+h, -h)$ is assigned to $\C T_2$.
\end{lem}
\begin{proof}
By assumption, since $\lambda$ assigns $(g,h,-g-h)$ to $\C T_1$ and $(g, h', -g-h')$ to $\C T_2$, we have $\lambda_4^{g+h} - \lambda_4^h < \mu_g$ and $\lambda_4^{g+h'} - \lambda_4^{h'} > \mu_g$ for all $h'\neq h$ (see proof of Lemma \ref{lem:inequalities}). Thus $\lambda_4^{g+h} - \lambda_4^h <  \lambda_4^{g+h'} - \lambda_4^{h'}$ for all $h'\neq h$. Now if $k=g + h$ then the result is tautological, so assume $k\neq g+h$. Let $h' = k -g$ so that $(-g, k, g - k) = (-g, g+h', -h')$. Now if $(-g, g+h', -h')$ is assigned to $\C T_2$ then $\lambda_4^{h'} - \lambda_4^{g + h'} > \mu_{-g}$. Since  $\lambda_4^{h'} - \lambda_4^{g + h'} = -(\lambda_4^{g + h'} - \lambda_4^{h'})$ and  $-(\lambda_4^{g + h'} - \lambda_4^{h'}) < -(\lambda_4^{g+h} - \lambda_4^h)$, we have $ -(\lambda_4^{g+h} - \lambda_4^h) = \lambda_4^h - \lambda_4^{g+h} > \mu_{-g}$. Thus, $\lambda(m_1((-g, g+h, -h)))>\lambda(m_2((-g, g+h, -h)))$, and consequently, $(-g, g+h, -h)$ is assigned to $\C T_2$.
\end{proof}

Next we describe a procedure for choosing $\lambda$ so that the matrix $A_\lambda$ has rank equal to the expected dimension.  This procedure will be illustrated below in Example \ref{exm:Z3}. For ease of notation, we will write $\nu$ for $\lambda_4 \in \mathbb{R}^{2t+1}$, so that $\nu_g = \lambda_4^g$ for all $g\in G$. In this notation, for a leaf labeling $\mathbf g = (g_1,g_2,-g_1-g_2)$,
$$ \lambda (m_1) < \lambda (m_2) \iff \nu_{g_1+g_2} - \nu_{g_2} < \mu_{g_1}.$$
This means $\mathbf{g} = (g_1,g_2,-g_1-g_2)$ is assigned to $\mathcal{T}_1$ if and only if
$\nu_{g_1+g_2} - \nu_{g_2} < \mu_{g_1},$
and $\mathbf{g}$ is assigned to $\mathcal{T}_2$ if and only if
$\nu_{g_1+g_2} - \nu_{g_2} > \mu_{g_1}.$

Choose $\nu$ such that $\nu_g \geq 0$ for all $g\in G$, with $\nu_0$ large enough and all $\nu_g$ small enough such that for all $h, h', k' \in G\setminus\{0\}$ we have $\nu_0 - \nu_h > \nu_{h'} - \nu_{k'} > \nu_h - \nu_0$.  Fix $g \in X$. By our choice of $\nu$, we can find $\mu_g$ such that $\nu_0 - \nu_{-g} > \mu_g$ and $\mu_g > \nu_{g+h} - \nu_h$ for all $h \in G$ with $h\neq -g$. Then for the consistent leaf-labelling $(g, -g, 0)$ we have $\nu_0 - \nu_{-g} > \mu _g$, so $(g, -g, 0)$ is assigned to $\C T_2$. For all $h\neq -g$ we have $\nu_{g + h} - \nu_h < \mu_g$ so the consistent leaf-labelling $(g, h, -g-h)$ is assigned to $\mathcal{T}_1$.

Next, consider $-g\not\in X$. In this case, for all $h\in G$ the consistent leaf-labelling $\mathbf{g'} = (-g, h, g -h)$ is assigned to $\mathcal{T}_1$ if and only if
$$\nu_{-g+h} - \nu_h < \mu_{-g},$$
and $\mathbf{g'}$ is assigned to $\mathcal{T}_2$ if and only if
$$\nu_{-g+h} - \nu_h > \mu_{-g}.$$

Choose $\mu_{-g} = -\mu_g$. Then for $(-g, 0, g)$ we have 
$$\nu_{-g} - \nu_{0} = -(\nu_0 - \nu_{-g}) < -\mu_{g} = \mu_{-g},$$
and thus $(-g, 0, g)$ is assigned to $\mathcal{T}_1$. Now, consider the inequality $\mu_g > \nu_{g+h} - \nu_{h}$, which holds for all $h \neq -g$. Write $h = -g + k$ with $k\neq 0$. Then we have $\mu_g > \nu_{k} - \nu_{-g+k}$ and therefore $\nu_{-g+k} - \nu_{k} > -\mu_g = \mu_{-g}$. Thus $(-g, k, g-k)$ is assigned to $\mathcal{T}_2$ for all $k\neq 0$. 

To summarize, at this point, we have found $\nu, \mu_g$ and $\mu_{-g}$ to give us the following assignments of the consistent leaf-labellings
 $(g,h,-g-h)$ and $(-g,h, g-h)$ for all $h\in G$:
 \begin{itemize}
      \item To $\mathcal{T}_1$ we have assigned  $(-g, 0, g)$ and $(g,h,-g-h)$ for all $h \neq -g$.
     \item To $\mathcal{T}_2$ we have assigned $(g,-g,0)$ and $(-g, h, g-h)$ for all $h \neq 0$.
 \end{itemize}

We repeat the above procedure for every $g\in X$. Then for each of the $t$ pairs $\{g, -g\} \subset G\setminus\{0\}$ we have $\ell+1$ consistent leaf-labellings assigned to $\mathcal{T}_1$ and $\ell+1$ consistent leaf-labellings assigned to $\mathcal{T}_2$. Now, for all $h\in G$ assign $(0,h, -h)$ to $\mathcal{T}_2$ so that we have $t(\ell+1) + (\ell + 1) = (t+1)(2t+1)$ consistent leaf-labellings assigned to $\mathcal{T}_2$ and $t(2t + 1)$ consistent leaf-labellings assigned to $\mathcal{T}_1$. Note that for $t \geq 2$ we have $2\ell = 4t < t(2t+1) < 2t(2t-1) = \ell(\ell-1)$.

We give a small example to illustrate.
\begin{example}\label{exm:Z3}
    Let $G = \mathbb{Z}/3\mathbb{Z} = \{0,1,2\}$, and let $X = \{1\}$. Pick $\nu_0 = 10$, $\nu_1 = 0$, and $\nu_2 = 2$. When $g=1$, for the consistent leaf-labellings $(1,2,0), (1,1,1)$, and $(1,0,2)$ respectively we have
    $$ \nu_0 - \nu_2 = 8 > \nu_2 - \nu_1 = 2 > \nu_1 - \nu_0 = -10.$$
    We choose $\mu_1 = 7$ so that $(1,2,0)$ is assigned to $\mathcal{T}_2$; and $(1,1,1)$ and $(1,0,2)$ are assigned to $\mathcal{T}_1$.

    When $g=2$, for the consistent leaf-labellings $(2,1,0), (2,2,2)$ and $(2,0,1)$ respectively we have
    $$ \nu_{0} - \nu_1 = 10 > \nu_{1} - \nu_{2} = -2 > \nu_{2} - \nu_{0} = -8.$$
    Setting $\mu_2 = -\mu_1 = -7$ gives us $(2,1,0)$ and $(2,2,2)$ assigned to $\mathcal{T}_2$ and $(2,0,1)$ assigned to $\mathcal{T}_1$.

    Finally, we choose $\mu_0 = -1$ so that $(0,0,0), (0,1,2)$, and $(0,2,1)$ are assigned to $\mathcal{T}_2$.
\end{example}

Returning to an arbitrary finite abelian group of odd order, we choose $\nu$ and $\mu$ as above (with $\mu_0$ negative) so that we have the assignments of leaf-labellings as described above. We will show that for this choice of $\lambda$ we have $\rank A_\lambda \geq 5\ell + 1$. Our strategy is to perform row and column operations on $A_\lambda$ in order to turn it into a block upper triangular matrix without changing the rank. 

We will denote the row of $A_\lambda$ corresponding to the parameter $w_i^g$ as $r_i^g$. Our first observation is that for any assignment of consistent leaf-labellings to $\C T_1$ and $\C T_2$ we have $r_1^g = r_5^g + r_6^g$ for all $g\in G$. The rows $r_1^g$ therefore do not contribute to the rank of $A_\lambda$, so we remove them. Note that this corresponds to the notion of a \emph{contracted semi-directed network}, which we do not introduce here (see \cite{gross2023dimensions} for further details). Next, perform column swaps so that all columns assigned to $\mathcal{T}_2$ are to the left of the columns assigned to $\mathcal{T}_1$. Observe that for all columns assigned to $\mathcal{T}_2$, the entry corresponding to $w_2^g$ is equal to the entry corresponding to $w_4^g$ for all $g \in G$ (where edge labels are as in Figure 2), and the entry corresponding to $w_5^g$ is $0$ for all $g\in G$. In terms of the tree topology, this is because in $\mathcal{T}_2$ the vertex between $e_2$ and $e_4$ has degree 2, and because $\mathcal{T}_2$ does not contain the edge $e_5$. Perform row swaps so that the corresponding edge order is $e_2, e_3, e_6, e_4, e_5$ from top to bottom. Next perform the row operations
$$r_4^g \rightarrow r_4^g - r_2^g$$
for all $g\in G$. This gives a block upper-triangular matrix of the form
\begin{equation} \label{eq:block-triangular-matrix}
 A_\lambda = 
\begin{blockarray}{ccc}
& \mathcal{T}_2 & \mathcal{T}_1 \\ [2pt]
\begin{block}{c(c|c)}
  r_2 & \qquad\qquad  & \qquad  \\ [0pt]
  r_3 &  A \qquad & * \\ [0pt]
  r_6 &  \quad\qquad &  \qquad \\ [0pt]
  \cline{2-3}
  && \\ [\dimexpr-\normalbaselineskip+10pt]
  r_4 - r_2 &  0 \qquad & B  \\ [0pt]
  r_5 & \qquad\qquad &  \\ [3pt]
\end{block}
\end{blockarray} \ .
 \end{equation}
 
\noindent It follows that $\rank A_\lambda \geq \rank A + \rank B$. The submatrix $A$ consists of all columns assigned to $\mathcal{T}_2$ and rows corresponding to $w_2^g, w_3^g,$ and $w_6^g$ for all $g \in G$. These are all the variables associated to edges in $\mathcal{T}_2$ that form a 3-star tree, and so we know $\rank A \leq \dim \mathcal{T}_2 = 3\ell + 1$. In Lemma \ref{lem:submatrixA}, we show that $\rank A = \dim \mathcal{T}_2 = 3\ell + 1$.  The submatrix $B$ consists of all columns assigned to $\mathcal{T}_1$ and rows $r_4^g - r_2^g$ and $r_5^g$ for all $g \in G$. The columns of $B$ are given by the consistent leaf-labellings assigned to $\mathcal{T}_1$ which are
$$ S = S_1\cup S_2 = \{(g,h,-g-h)\ |\ g\in X, h\neq -g\}\cup\{(-g,0,g)\ |\ g\in X\},$$
so $|S_1| = 2t(t) = \frac{\ell}{2}\ell$ and $|S_2|=t = \frac{\ell}{2}$.  In the next example, we illustrate the block triangular matrix in \eqref{eq:block-triangular-matrix} for $G = \mathbb{Z}/{3 \mathbb{Z}}$.
\begin{example}
    Here we continue Example \ref{exm:Z3}. With $G = \mathbb{Z}/3\mathbb{Z}$ and assignments of consistent leaf-labellings as in the example, after all row operations we have the following matrix.
     \[ A_\lambda = 
\begin{blockarray}{cccccccccc}
& \mathcal{T}_2 & \mathcal{T}_2 & \mathcal{T}_2 & \mathcal{T}_2 & \mathcal{T}_2 & \mathcal{T}_2 & \mathcal{T}_1 & \mathcal{T}_1 & \mathcal{T}_1 \\ [2pt]
& 000 & 012 & 021 & 120 & 210 & 222 & 111 & 102 & 201 \\ [2pt]
\begin{block}{c(cccccc|ccc)}
  w_2^0 & 1 & 0 & 0 & 0 & 0 & 0 & 0 & 1 & 1 \\ [3pt]
  w_2^1 & 0 & 1 & 0 & 0 & 1 & 0 & 1 & 0 & 0 \\ [3pt]
  w_2^2 & 0 & 0 & 1 & 1 & 0 & 1 & 0 & 0 & 0 \\ [3pt]
  w_3^0 & 1 & 0 & 0 & 1 & 1 & 0 & 0 & 0 & 0  \\ [3pt]
  w_3^1 & 0 & 0 & 1 & 0 & 0 & 0 & 1 & 0 & 1  \\ [3pt]
  w_3^2 & 0 & 1 & 0 & 0 & 0 & 1 & 0 & 1 & 0  \\ [3pt]
  w_6^0 & 1 & 1 & 1 & 0 & 0 & 0 & 0 & 0 & 0  \\ [3pt]
  w_6^1 & 0 & 0 & 0 & 1 & 0 & 0 & 0 & 0 & 0  \\ [3pt]
  w_6^2 & 0 & 0 & 0 & 0 & 1 & 1 & 0 & 0 & 0  \\ [2pt]
  \cline{2-10}
  &&&&&&&&& \\ [\dimexpr-\normalbaselineskip+3pt]
  w_4^0 & 0 & 0 & 0 & 0 & 0 & 0 & 0 & -1 & -1  \\ [3pt]
  w_4^1 & 0 & 0 & 0 & 0 & 0 & 0 & -1 & 1 & 0  \\ [3pt]
  w_4^2 & 0 & 0 & 0 & 0 & 0 & 0 & 1 & 0 & 1  \\ [3pt]
  w_5^0 & 0 & 0 & 0 & 0 & 0 & 0 & 0 & 0 & 0  \\ [3pt]
  w_5^1 & 0 & 0 & 0 & 0 & 0 & 0 & 1 & 1 & 0  \\ [3pt]
  w_5^2 & 0 & 0 & 0 & 0 & 0 & 0 & 0 & 0 & 1  \\ [3pt]
\end{block}
\end{blockarray}
 \]

 Note that in this case the dimension of the space containing the variety is $(\ell + 1)^2 = 9$, which is less than the expected dimension of $5\ell +1 = 11$. Therefore the expected dimension cannot be reached, and indeed, explicit computation shows that the variety has dimension $9$. Here we can see explicitly that the submatrices $A$ and $B$ have rank strictly less than as described in Lemmas \ref{lem:submatrixA} and \ref{lem:submatrixB}.
\end{example}

In the next two lemmas, we give the rank of the submatrix $A$ and a lower bound on the rank of the submatrix $B$.

\begin{lem}\label{lem:submatrixA}
Let $G$ be an abelian group with $|G|=\ell+1$ odd and $|G| \geq 5$, and $A_\lambda$ as above in \eqref{eq:block-triangular-matrix}. Then $\rank A = 3\ell + 1$.
\end{lem}
\begin{proof}
    The submatrix $A$ consists of columns corresponding to consistent leaf-labellings assigned to $\mathcal{T}_2$, and rows corresponding to the variables $w_2^g, w_3^g$, and $w_6^g$ for all $g\in G$. This submatrix also appears in the corresponding matrix of exponents for the $3$-star tree with edges $e_2, e_3$, and $e_6$ as in \Cref{fig:3star} (possibly after performing some column swaps). 

    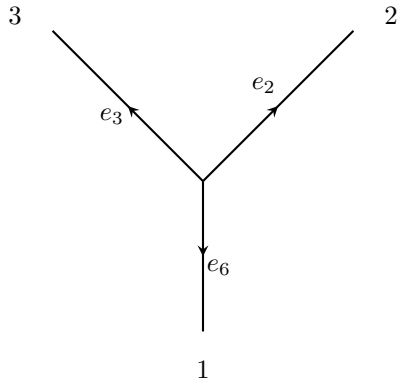
\begin{figure}[h!]
        \centering
        \begin{tikzpicture}
        [every node/.style={inner sep=0pt},
                            every path/.style={thick},   
        decoration={markings, 
            mark= at position 0.5 with {\arrow{stealth}}
            }
        ]
        \draw[postaction={decorate}] (0,0) -- (-2,2) node[midway,below=4, left=1] {$e_{3}$};
        \draw[postaction={decorate}] (0,0)  -- (2 ,2)  node[midway,above=8,left=0.5] {$e_{2}$};
        \draw[postaction={decorate}] (0,0) -- (0,-2) node[midway,below=4, right=1] {$e_{6}$};
        
        \node (v3) at (0,-2.5) {$1$};
        \node (v3) at (2.5,2.2) {$2$};
        \node (v3) at (-2.5,2.2) {$3$};
        
        \end{tikzpicture}
        \caption{A 3-star tree with taxa labels and edge labels}
        \label{fig:3star}
    \end{figure}

    Denote by $K_3$ the matrix of exponents corresponding to the general group-based model of $G$ on the $3$-star tree. The dimension of the variety corresponding to this model is $3\ell +1$ \cite[Lemma~4.1]{gross2023dimensions}. Since the parameterization of this model is monomial, $K_3$ has rank equal to the dimension of the model.
     
    Let $K_{(g_1, g_2, g_3)}$ be the column of $K_3$ corresponding to the consistent leaf-labelling $(g_1, g_2, g_3)$, so that we have
    $$K_{(g_1, g_2, g_3)} = E_6^{g_1} + E_2^{g_2} + E_3^{g_3}.$$
    We will show that the columns of $K_3$ corresponding to consistent leaf-labellings that we have assigned to $\mathcal{T}_1$ can be written as linear combinations of columns corresponding to consistent leaf-labellings we have assigned to $\mathcal{T}_2$, thereby showing that $\rank A = \rank K_3 = 3\ell + 1$.

    First, consider the consistent leaf-labelling $(-g,0,g)$. The reader can check that for any $h\in X$ with $h\neq g$ we have
    $$K_{(-g,0,g)} = K_{(0,0,0)} + K_{(-g, h, g-h)} + K_{(-h,h-g,g)} - K_{(-h,h,0)} - K_{(0,h-g,g-h)},$$
    and all terms on the right hand side are from consistent leaf-labellings that are assigned to $\mathcal{T}_2$. Note that since $|G|\geq 5$, at least one such $h$ exists.
    
    Next consider the consistent leaf-labelling $(g, h, -g-h)$ for $g \in X$ and $h\neq -g$. We consider two cases separately:\\
    
    \emph{Case 1:} $h \not\in X$.
    For this case we may write
    \begin{equation}\label{eq:linear-dependence}
    K_{(g, h, -g-h)} = K_{(g, -g, 0)} + K_{(-g, h, g-h)} + K_{(h, g , -g-h)} - K_{(-g, g, 0)} - K_{(h, -g, g-h)}.
    \end{equation}
    
    \emph{Case 2:} $h \in X$.
    We break this down into two further cases. First, suppose $-g-h \in X$. Then we have
    \begin{equation}\label{eqn:Ks}
        K_{(g, h, -g-h)} = K_{(g, -g, 0)} + K_{(0, h, -h)} + K_{(g+h, 0 , -g-h)} - K_{(0, 0, 0)} - K_{(g+h, -g, -h)},
    \end{equation}
    where we are using that although $(g+h, 0 , -g-h)$ is assigned to $\C T_1$, the column $K_{(g+h, 0 , -g-h)}$ is linearly dependent on columns corresponding to leaf-labellings assigned to $\mathcal{T}_2$ by the first part of the proof. Therefore we can substitute the relation from \eqref{eq:linear-dependence} into \eqref{eqn:Ks} to obtain $K_{(g, h, -g-h)}$ as a linear combination of columns assigned to $\C T_2$.  On the other hand, if $-g-h\not\in X$ then
    $$K_{(g, h, -g-h)} = K_{(g, -g, 0)} + K_{(-g-h, h, g)} + K_{(0, g+h , -g-h)} - K_{(0, -g, g)} - K_{(-g-h, g+h, 0)}.$$  

\end{proof}

\begin{remark}
    Observe that the relations used in the proof above correspond to the cubic binomials in the ideal for the 3-star tree, as described in \cite{sturmfels2005toric}.
    \end{remark}

\begin{lem}\label{lem:submatrixB}
Let $G$ be a finite abelian group with $|G|=\ell+1$ odd and $|G|\geq 7$, and $A_\lambda$ as above in \eqref{eq:block-triangular-matrix}. Then $\rank B \geq 2\ell$.
\end{lem}
\begin{proof}
The columns of $B$ correspond to the consistent leaf-labellings we assign to $\C T_1$, that is, those in the set $S$, where
$$ S = S_1\cup S_2 = \{(g,h,-g-h)\ |\ g\in X, h\neq -g\}\cup\{(-g,0,g)\ |\ g\in X\}.$$
The rows of $B$ correspond to the variables $w_5^g$ and $w_4^g$, where we have performed the row operation $r_4^g \rightarrow r_4^g - r_2^g$, for all $g\in G$. For the leaf-labelling $(g_1, g_2, g_3)$, we have $m_1(g_1,g_2,g_3) = w_1^{g_1}w_2^{g_2}w_3^{g_3}w_4^{g_1 + g_2}w_5^{g_1}$. Thus, the column of $B$ corresponding to that leaf labelling is the vector in $\mathbb{C}^{2\ell+2}$ given by
$$B_{(g_1, g_2, g_3)} = E_5^{g_1} + E_4^{g_1 + g_2} - E_4^{g_2}.$$
To prove the result, it is sufficient to find a subset $\C{B}\subset S$ of $2\ell$ linearly independent columns.

First consider the columns corresponding to consistent leaf-labellings in $S_2$. Here each column is given by
$$B_{(-g,0,g)} = E_5^{-g} + E_4^{-g} - E_4^0,$$
for $g\in X$. Since $-g\not\in X$, the column $B_{(-g,0,g)}$ is the only column of $B$ with a non-zero component in the basis vector $E_5^{-g}$, and is therefore linearly independent of all other columns.

Next we consider columns coming from consistent leaf-labellings in $S_1$. For each $g\in X$ consider the leaf-labellings $(g,0,-g),(g,g,-2g),$ and $(g, h_g, -g-h_g)$, where for each $g\in X$ the element $h_g$ is chosen from $G\setminus(X\cup\{0\})$ with the conditions that $h_g \neq -g, 2g$, and if $g\neq g'$ then $h_g\neq h_{g'}$ (this is possible by Lemma \ref{lem:hg}). We claim that the set 
$$\C{B}=\{B_{(-g,0,g)}\ |\ g\in X\}\cup\{B_{(g,0,-g)},B_{(g,g,-2g)},B_{(g, h_g, -g-h_g)}\ |\ g\in X\}$$
is linearly independent. Let
$$V_g = \text{span}_{\mathbb{C}}\{B_{(g,0,-g)}, B_{(g,g,-2g)}, B_{(g, h_g, -g-h_g)}\}.$$
By the choice of $h_g$, $\dim V_g = 3$ for all $g \in X$.  Thus, to prove the claim, it is sufficient to show that $V_g\cap V_{g'}=\{0\}$ for all $g, g' \in X$ with $g\neq g'$.

First, by considering the basis vectors $E_5^g$ and $E_5^{g'}$, we must have that $V_g\cap V_{g'} \subset W$ where $W = \text{span}_{\mathbb{C}}\{E_4^h\ |\ h \in G\}$, so that $V_g\cap V_{g'} = (V_g\cap W)\cap(V_{g'}\cap W)$. Now $V_g\cap W$ is spanned by the vectors $B_{(g,0, -g)} - B_{(g,g,-2g)}$ and $B_{(g,0, -g)} - B_{(g,h_g,-g-h_g)}$, where
\begin{align*}
B_{(g,0, -g)} - B_{(g,g,-2g)} &= 2E_4^{g} - E_4^{2g} - E_4^0, \\
B_{(g,0, -g)} - B_{(g,h_g,-g-h_g)} &= E_4^{g} - E_4^{g+h_g} + E_4^{h_g} - E_4^0.
\end{align*}
An arbitrary element of $v\in V_g\cap W$ is then given by 
\begin{align*}
   v &= \alpha(B_{(g,0, -g)} - B_{(g,g,-2g)}) + \beta(B_{(g,0, -g)} - B_{(g,h_g,-g-h_g)}) \\
     &= (2\alpha + \beta)E_4^g -\alpha E_4^{2g} + \beta E_4^{h_g} - \beta E_4^{g+h_g} -(\alpha +\beta)E_4^0,
\end{align*}
for $\alpha, \beta \in \mathbb{C}$. Observe that from the choice of $h_g$, the basis vectors appearing in this expression are distinct. Now suppose that for $g'\in X$ with $g\neq g'$, we have another such element $v'\in V_{g'}\cap W$ with coefficients $\alpha'$ and $\beta'$. We will show that if $v=v'$ then $v = v' = 0$ so that $(V_g\cap W)\cap(V_{g'}\cap W) = \{0\}$ as claimed. We have
\begin{align}\label{eqn:sim}
    v &= (2\alpha + \beta)E_4^g -\alpha E_4^{2g} + \beta E_4^{h_g} - \beta E_4^{g+h_g} -(\alpha +\beta)E_4^0, \\
    v' &= (2\alpha' + \beta')E_4^{g'} -\alpha' E_4^{2g'} + \beta' E_4^{h_{g'}} - \beta' E_4^{g'+h_{g'}} -(\alpha' +\beta')E_4^0.
\end{align}
Now suppose $v=v'$. By examining the coefficients of $E_4^0$ we must have 
\begin{equation}\label{eqn:alphabeta}
\alpha + \beta = \alpha' + \beta'.
\end{equation}
Next consider the coefficient of $E_4^{g'}$. Either $E_4^{g'}$ appears in the expression for $v$ or $2\alpha' + \beta' = 0$. For the former, since $g'\neq g, 0$ and $g' \in X$ so $g' \neq h_g$, we must have either $g'=2g$ or $g' = g + h_g$. We consider these three cases separately.\\

\emph{ Case 1:} $g'=2g$.
By equating coefficients of $E_4^{g'} = E_4^{2g}$ we have $2\alpha' + \beta' = -\alpha$, and substituting in to equation (\ref{eqn:alphabeta}) gives $\beta = 3\alpha' + 2\beta'$. Next consider the coefficient of $E_4^{h_{g'}}$. Either this is zero (i.e. $\beta' = 0)$ or $E_4^{h_{g'}}$ appears in the expression for $v$. 

If $\beta'\neq 0$ then since $h_{g'} \neq g, h_g, 0, $ or $g'=2g$, we must have $h_{g'} = g + h_g,$ and therefore $\beta' = -\beta$, from which it follows by \Cref{eqn:alphabeta} that $\alpha + 2\beta = \alpha'$. Now since the coefficient of $E_4^{g'+h_{g'}}$ is not zero we must have either $g'+h_{g'} = g$ or $h_g$. If $g'+h_{g'} = h_g$ then substituting in $h_{g'} = g + h_g$ gives $g' + g +h_g = h_g$ and therefore $g = -g'$, a contradiction. We therefore must have $g'+h_{g'} = g$ so then $-\beta' = 2\alpha + \beta$, i.e., $\alpha = 0$. But now we have simultaneous equations $2\beta = \alpha'$ and $3\beta = \alpha'$, so $\alpha' = \beta = 0$.

Now, if $\beta' = 0$ then we have $\alpha = -2\alpha'$ and $\beta = 3\alpha'$. Substituting these values into equation (\ref{eqn:sim}) and setting $v=v'$ gives
$$-\alpha'E_4^g + 2\alpha'E_4^{2g} + 3\alpha'E_4^{h_g} - 3\alpha'E_4^{g+h_g} = 2\alpha'E_4^{g'} - \alpha'E_4^{2g'},$$
so we conclude that $\alpha' = 0$, and therefore $v=v' = 0$.

 \emph{Case 2: }$g'=g + h_g$.
By equating coefficients of $E_4^{g'} = E_4^{g+h_g}$ we have $2\alpha' + \beta' = -\beta$. Next consider the coefficient of $E_4^{h_{g'}}$. As before, either this is zero or $E_4^{h_{g'}}$ appears in the expression for $v$. 

If $\beta' \neq 0$ then we must have $h_{g'} = 2g$, so that $\beta' = -\alpha$ and therefore $\beta = \alpha' + 2\beta'$ by \Cref{eqn:alphabeta}, and then $3\alpha' + 3\beta' = 0$. Now consider the coefficient of $E_4^{g' + h_{g'}}$. Either $g' + h_{g'} = g$ or $g'+h_{g'} = h_g$. If $g' + h_{g'} = g$ then $g+h_g + h_{g'} = g$ so $h_g$ = $-h_{g'}$, but this contradicts how $h_g$ and $h_{g'}$ were chosen; both $h_g$ and $h_{g'} \in G\setminus (X\cup \{0\})$, but then by definition of $X$ we must have $-h_g$ and $-h_{g'}\in X$. Thus we must have $g'+h_{g'} = h_g$, and so $\beta = -\beta'$. Now we have $\beta' = -\alpha = -\beta$ so by equation \ref{eqn:alphabeta} we have $0=\alpha' + 3\beta'$. Solving this simultaneously with $3\alpha' + 3\beta' = 0$ gives $v = v' = 0$.

If $\beta' = 0$ we have $2\alpha' = -\beta$ and so by \Cref{eqn:alphabeta} $\alpha = 3\alpha'$. Substituting into equation (\ref{eqn:sim}) and setting $v=v'$ gives us $\alpha' = 0$ as in case 1, and therefore $v = v' = 0$.

\emph{Case 3:} $2\alpha' + \beta' = 0$. 
In this case we have $\beta' = -2\alpha'$ so $\alpha' = -(\alpha + \beta)$. This gives us
$$v' = (\alpha + \beta)E_4^{2g'} + 2(\alpha + \beta)E_4^{h_{g'}} - 2(\alpha + \beta)E_4^{g'+h_{g'}} - (\alpha + \beta)E_4^0.$$
In particular, if $v'\neq 0$ then $\alpha + \beta \neq 0$ so that $v'$ has 4 linearly independent non-zero terms. We have
$$v = (2\alpha + \beta)E_4^g - \alpha E_4^{2g} + \beta E_4^{h_g} - \beta E_4^{g+h_g} - (\alpha + \beta)E_4^0,$$
so if $v=v'$ we must have exactly four non-zero terms in the expression for $v$. This is only possible if $\alpha =0$ or $2\alpha + \beta = 0$. If $\alpha = 0$ then we have
\begin{align*}
    v &= \beta E_4^{g} + \beta E_4^{h_g} - \beta E_4^{g + h_g} - \beta E_4^0, \\
    v' &= \beta E_4^{2g'} +2\beta E_4^{h_{g'}} - 2\beta E_4^{g' + h_{g'}} - \beta E_4^0.
\end{align*}
If $v=v'$ then by equating coefficients we must have $\beta = 0$ and thus $v = v' = 0$. On the other hand, if $2\alpha + \beta = 0$ we get
\begin{align*}
    v &= -\alpha E_4^{2g} - 2\alpha E_4^{h_g} + 2\alpha E_4^{g + h_g} + \alpha E_4^0, \\
    v' &= -\alpha E_4^{2g'} -2\alpha E_4^{h_{g'}} + 2\alpha E_4^{g' + h_{g'}} + \alpha E_4^0.
\end{align*}
Examining the coefficient of $E_4^{2g}$ we see that since $g\neq g'$ and $g \neq 0$ we must have $-\alpha = \pm 2\alpha$, so $\alpha = 0$ and $v = v' = 0$.
\end{proof}

Putting together the results of this section and \cite[Proposition~4.2]{gross2023dimensions}, we get Theorem \ref{thm:main}, which we restate here for convenience. By direct calculation (see \Cref{tab:rank-samples-5}), we have that for $G=\mathbb{Z}/5\mathbb{Z}$, the dimension of $V^G_\mathcal{N}$ is 21, the expected dimension, so we include this result in the statement.

\begingroup
\def\thetheorem{\ref{thm:main}}
\begin{theorem}
Let $G$ be a finite abelian group of odd order $\ell+1 \geq 5$, and let $\mathcal{N}$ be the 3-sunlet network under the general group-based model given by $G$, with corresponding phylogenetic network variety $V_\mathcal{N}^G$. Then the affine dimension of $V_\mathcal{N}^G$ is given by 
$$\dim V_\mathcal{N}^G = 5\ell + 1.$$
\end{theorem}
\addtocounter{theorem}{-1}
\endgroup
\begin{proof}
    For $G=\mathbb{Z}/5\mathbb{Z}$, the result is acheived computationally by using random sampling to find a weight vector $\lambda$ for which $A_\lambda$ has rank $5\ell + 1 = 21$ (\Cref{tab:rank-samples-5}). For all other $G$, choose $\lambda$ as described in this section. Through row operations we can transform the matrix $A_\lambda$ into a block upper triangular matrix of the form
    \begin{equation*}
	A_\lambda = \left(
	\begin{array}{cc}
	A& *\\ 
	0& B
	\end{array}\right),
    \end{equation*}
    so that $\rank A_\lambda \geq \rank A + \rank B$. By Lemma \ref{lem:submatrixA}, $\rank A = 3\ell + 1$, and, by Lemma \ref{lem:submatrixB},  $\rank B \geq 2\ell$. Thus, we obtain $\rank A_\lambda \geq 5\ell + 1$, and so the affine dimension of $V_\mathcal{N}^G$ is at least $5\ell + 1$. On the other hand, Lemma \ref{lem:upper} says that $\dim V_\mathcal{N}^G$ is at most $5\ell + 1$.  
\end{proof}
As discussed at the beginning of this work, the $3$-sunlet was the only sunlet where a dimension formula was not given in \cite{gross2023dimensions}. Since level-1 phylogenetic networks can be broken down into trees and sunlet networks, for the case when $G$ is an abelian group of odd order at least $5$ and $\mathcal{N}$ is a level-1 phylogenetic network, we can now give a full dimension result.

\begin{theorem}\label{thm:level-1}
Let $\mathcal{N}$ be a level-1 phylgenetic network with $n$ leaves, $m$ edges, and $c$ cycles. Let $G$ be a finite abelian group of odd order $\ell + 1 \geq 5$. Then the variety corresponding to $\C N$ under the general group-based model for $G$, denoted $V_{\mathcal{N}}^G$, has dimension $\ell(m-c) + 1$.
\end{theorem}
\begin{proof}
    Following \cite[Theorem~1.1]{gross2023dimensions}, we prove the result by induction on the number of cut edges of $\mathcal{N}$. If $\mathcal{N}$ has no cut edges, then it is either the 3-star tree or a sunlet network. If $\mathcal{N}$ is the 3-star tree, then it has no cycles and $\dim V_{\mathcal{N}}^G =3\ell + 1$ \cite[Lemma~4.1]{gross2023dimensions}. If $\mathcal{N}$ is a 3-sunlet network then it has a single cycle and 6 edges, and $\dim V_{\mathcal{N}}^G = 5\ell + 1$ by Theorem \ref{thm:main}. If $\mathcal{N}$ is an $n$-sunlet network with $n > 3$ then it has a single cycle and $2n$ edges, and $\dim V_{\mathcal{N}}^G = \ell(2n-1) + 1$ \cite[Theorem~4.7]{gross2023dimensions}. Thus in all cases the result holds.

    For the induction step, suppose that $\mathcal{N}$ is a level-1 phylogenetic network with a cut edge $e$, $m$ edges, and $c$ cycles. Let $\mathcal{N}_1$ and $\mathcal{N}_2$ be the networks obtained by cutting $\mathcal{N}$ at $e$ and observe that both $\mathcal{N}_1$ and $\mathcal{N}_2$ have strictly fewer cut edges than $\mathcal{N}$. For $i=1,2$ let $m_i$ and $c_i$ be the number of edges and cycles respectively of $\mathcal{N}_i$, then by induction we have $\dim V_{\mathcal{N}_i}^G = \ell(m_i - c_i) + 1$. Next we have that the ideal $I_{\mathcal{N}}^G$ defining $V_{\mathcal{N}_i}^G$ is given by the toric fiber product of $I_{\mathcal{N}_1}^G$ and $I_{\mathcal{N}_2}^G$ \cite[Remark~3.3]{cummings2021invariants}. Using \cite[Corollary~3.4]{gross2023dimensions} we obtain
    \begin{align*}
        \dim V_{\mathcal{N}}^G &= \dim V_{\mathcal{N}_1}^G + \dim V_{\mathcal{N}_2}^G - (\ell + 1) \\
            &= \ell(m_1 - c_1) + 1 + \ell(m_2 - c_2) + 1 -(\ell + 1) \\
            &=\ell(m_1 + m _2 - 1 - (c_1 + c_2)) + 1 \\
            &=\ell(m - c) + 1.
    \end{align*}
    
\end{proof}

The dimension results in \cite{gross2023dimensions} were used to prove identifiability statements for level-1 phylogenetic networks under group-based models. Now that we have the dimension result for 3-sunlet networks, we can remove the \lq triangle-free\rq\ restriction placed on those statements for the general-group based model for $G$, when $G$ is a finite abelian group of odd order at least 5. The proof of the following result is identical to \cite[Proposition~6.9]{gross2023dimensions}. First, recall the following two definitions. We say that two phylogenetic networks $\mathcal{N}_1$ and $\mathcal{N}_2$ are \emph{distinguishable over $G$} if $V_{\mathcal{N}_1}^G \not\subset V_{\mathcal{N}_2}^G$ and $V_{\mathcal{N}_2}^G \not\subset V_{\mathcal{N}_1}^G$. Given a level-1 phylogenetic network $\mathcal{N}$ on $n$ leaves and a subset $A$ of the leaf-set $[n]$, the \emph{network restricted to $A$}, denoted $\mathcal{N}|_A$, is the level-1 phylogenetic network obtained from $\mathcal{N}$ by removing all edges and vertices that do not lie on any path between two leaves in $A$, and suppressing any resulting vertices of degree 2 (except the root vertex). See \cite[Definition~4.1]{gross2018distinguishing} for a full definition.

\begin{proposition}\label{prop:identifiability}
Let $\mathcal{N}_1$ and $\mathcal{N}_2$ be two level-1 phylogenetic networks on $n$ leaves and both with exactly $c$ cycles. Let $G$ be a finite abelian group of odd order $\geq 5$. If there exists a subset $A\subset [n]$ such that either
\begin{enumerate}
    \item $\mathcal{N}_1|_A$ and $\mathcal{N}_2|_A$ are level-1 phylogenetic networks with distinct numbers of cycles; or
    \item $\mathcal{N}_1|_A$ is a tree and $\mathcal{N}_2|_A$ is a level-1 phylogenetic network (i.e. with at least 1 cycle); or
    \item $\mathcal{N}_1|_A$ and $\mathcal{N}_2|_A$ are distinct trees;
\end{enumerate}
then $\mathcal{N}_1$ and $\mathcal{N}_2$ are distinguishable over $G$. \qed
\end{proposition}
 
\section{Experimental Results}
\label{sec:experiments}

Theorem \ref{thm:main} confirms that for odd order groups, the general group-based model on a 3-sunlet has the expected dimension $5\ell + 1$ (except $G=\mathbb{Z}/3\mathbb{Z}$). In this section, we investigate the dimension for small finite abelian groups, and whilst an analogous construction of $\lambda$ for even order groups does not give us a maximal rank $A_\lambda$, through experiments we find that the expected dimension is obtained in all cases once $G$ is sufficiently large. The code we used to perform these calculations was written in Julia \cite{bezanson2017julia} and available to download at \texttt{https://github.com/shelbycox/3-Sunlet}.

\subsection{Sampling Methods}

We use the hyperplane description from \Cref{subsec: hyperplanes} to compute the possible matrices $A_\lambda$ (and their ranks) that can appear for the groups $\ZZ/3\ZZ$, $\ZZ/4\ZZ$, $\ZZ/2\ZZ \times \ZZ/2\ZZ$, and $\ZZ/5\ZZ$. For each of the possible $2^{\ell(\ell + 1) + 1}$ regions, we use \texttt{OSCAR} \cite{OSCAR} to test whether the region is full dimensional. We then use random sampling of points $p = (\mathbf{\mu}, \mathbf{\lambda_4})$ in $(-.5,.5)^{2(\ell + 1)}$ to obtain exactly one point in each region. For other small groups mentioned in this section, we used only random sampling with $2^{32}$ samples for each group to obtain the results. In each case, we retain at most one point for each region.

\subsection{Chamber counts for small groups} 
\label{subsec:sampling}
\Cref{tab:rank-samples-3,,tab:rank-samples-4,,tab:rank-samples-22,,tab:rank-samples-5,,tab:rank-samples-6} are the result of the computations described above for the groups $\ZZ/3\ZZ$, $\ZZ/4\ZZ$, $\ZZ/2\ZZ \times \ZZ/2\ZZ$, $\mathbb{Z}/5\mathbb{Z}$, and $\mathbb{Z}/6\mathbb{Z}$ respectively. For the first four groups listed, we confirmed that these are exactly the regions of the hyperplane arrangement. For $\ZZ/6\ZZ$, the data in the table is the result of sampling $2^{32}$ points and then retaining no more than one point per region. In addition, we include an illustration of relationships between chambers for $\ZZ/3\ZZ$ in Figure \ref{fig:Z3-poset-graph}. We make some observations.

\begin{remark}
    As observed in \Cref{lem:minimalRankChambers}, in all cases we have exactly 4 chambers of lowest rank, which is equal to the dimension of the 3-star tree, $3\ell + 1$. Two of the chambers correspond to when all leaf-labellings are assigned to either $\mathcal{T}_1$ or $\mathcal{T}_2$. Weight vectors $\lambda$ with $\lambda_5^g$ sufficiently large and $\lambda_6^g =0$ for all $g\in G$, or vice-versa, lie in these chambers. The remaining two chambers correspond to when all leaf-labellings $(0,g,-g)$ are assigned to $\mathcal{T}_1$ and all others assigned to $\mathcal{T}_2$; and when all leaf-labellings $(0,g,-g)$ are assigned to $\mathcal{T}_2$ and all others assigned to $\mathcal{T}_1$. These chambers can be reached by taking the previous weight vector $\lambda$ and swapping the values of $\lambda_5^0$ and $\lambda_6^0$.

    For $G=\mathbb{Z}/3\mathbb{Z}$, these chambers can be seen in \Cref{fig:Z3-poset-graph}, drawn in squares and highlighted in blue at the very top and very bottom of the diagram. The shortest path in the poset between $\C T_1$ and $\C T_2$ has length 7, meaning that 7 hyperplanes need to be crossed to reach one from the other. 
\end{remark}
\begin{remark}
    When $G = \ZZ/3\ZZ$, most chambers have rank 9, which is the maximum possible, because the space containing the corresponding variety is $\mathbb{C}^9$. However, for $\ZZ/4\ZZ$ the maximum rank, 16, is achieved by only $4.8\%$ of chambers. For $\ZZ/2\ZZ \times \ZZ/2\ZZ$ the maximum rank is 15, which is achieved by only $20.6\%$ of chambers. For both $\ZZ/4\ZZ$ and $\ZZ/2\ZZ \times \ZZ/2\ZZ$ rank 14 chambers are observed the most ($37.1\%$ and $45.4\%$).
\end{remark}

\begin{remark}
    From \Cref{tab:rank-samples-4} and \Cref{tab:rank-samples-22}, we observe that the number of chambers in the hyperplane arrangement is the same for $\mathbb{Z}/4\mathbb{Z}$ and $\ZZ/2\ZZ \times \ZZ/2\ZZ$. We speculate that this is true in general for groups of the same order. However, we also observe that the distribution of ranks for these chambers differs. Thus, the distribution of ranks can depend on the structure of $G$, and not just $|G|$.
\end{remark}

\begin{figure}
    \centering
    \includegraphics[angle=90, scale=.45]{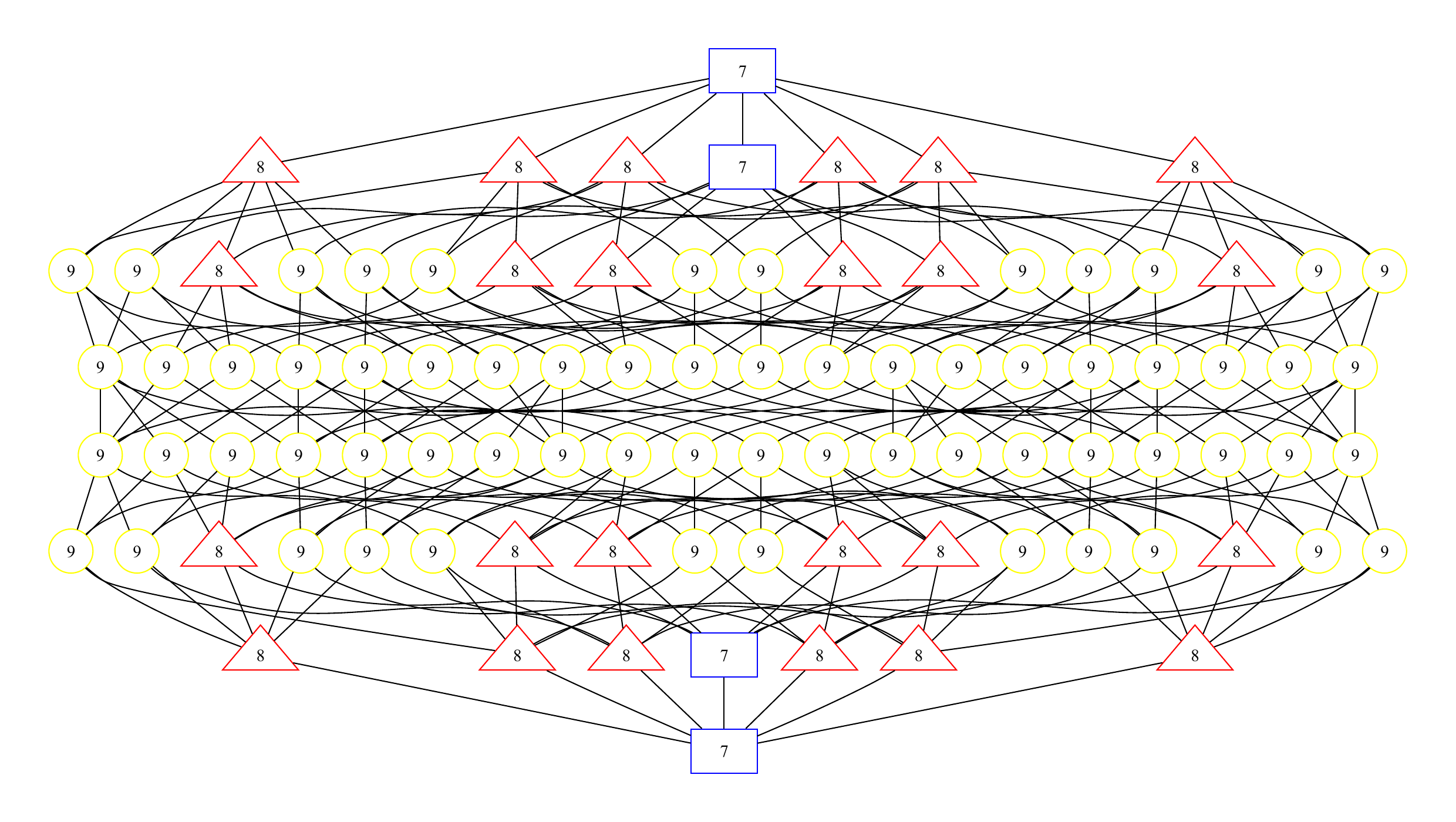}
    \caption{Poset of $\C{H}(\ZZ/3\ZZ)$, graded by distance (right to left) from the starting chamber, $\C T_1$ (at the bottom of the poset or far right of the diagram). The number in each node is the rank of the corresponding chamber. \\\textbf{Key}: blue/square - $\rank A_\lambda = 7$, red/triangle - $\rank A_\lambda = 8$, yellow/circle - $\rank A_\lambda = 9$.}
    \label{fig:Z3-poset-graph}
\end{figure}

\begin{table}[!h]
    \centering
    \begin{tabular}{c|ccc}
        rank & 7 & 8 & 9 \\
        \hline
        \# chambers & 4 & 24 & 64 \\
        \% chambers & 4.3\% & 26.1\% & 69.6\%
    \end{tabular}
    \caption{The number of samples for each $A_\lambda$ rank when $G = \ZZ/3\ZZ$.}
    \label{tab:rank-samples-3}
\end{table}

\begin{table}[!h]
    \centering
    \begin{tabular}{c|ccccccc}
        rank & 10 & 11 & 12 & 13 & 14 & 15 & 16 \\
        \hline
        \# chambers & 4 & 48 & 180 & 496 & 864 & 624 & 112 \\
        \% chambers & .2\% & 2.06\% & 7.7\% & 21.3\% & 37.1\% & 26.8\% & 4.8\%
    \end{tabular}
    \caption{The number of samples for each $A_\lambda$ rank when $G = \ZZ/4\ZZ$.}
    \label{tab:rank-samples-4}
\end{table}

\begin{table}[!h]
    \centering
    \begin{tabular}{c|ccccccc}
        rank & 10 & 11 & 12 & 13 & 14 & 15 & 16 \\
        \hline
        \# chambers & 4 & 48 & 156 & 584 & 1056 & 480 & 0 \\
        \% chambers & .2\% & 2.06\% & 6.7\% & 25.1\% & 45.4\% & 20.6\% & 0\%
    \end{tabular}
    \caption{The number of samples for each $A_\lambda$ rank when $G = \ZZ/2\ZZ \times \ZZ/2\ZZ$.}
    \label{tab:rank-samples-22}
\end{table}

\begin{table}[!h]
    \centering
    \begin{tabular}{c|ccccccccc}
        rank & 13 & 14 & 15 & 16 & 17 & 18 & 19 & 20 & 21 \\
        \hline
        \# chambers & 4 & 80 & 560 & 2,160 & 5,228 & 11,520 & 27,960 & 41,360 & 24,480 \\
        \% chambers & .004\% & .071\% & .494\% & 1.906\% & 4.612\% & 10.163\% & 24.667\% & 36.488\% & 21.596\%
    \end{tabular}
    \caption{The number of samples for each $A_\lambda$ rank when $G = \ZZ/5\ZZ$.}
    \label{tab:rank-samples-5}
\end{table}

\begin{table}[!h]
    \centering
    \begin{tabular}{c|cccccc}
        rank & 16 & 17 & 18 & 19 & 20 & 21 \\
        \hline
        \# chambers & 4 & 120 & 1296 & 7,180 & 26,576 & 79,156 \\
        \% chambers & .00004\% & .00129\% & .01391\% & .07707\% & .28528\% & .84969\%  \\ [3pt]
        rank & 22 & 23 & 24 & 25 & 26 & \\
        \hline
        \# chambers & 229,069 & 742,458 & 2,148,510 & 3,606,334 & 2,475,117 & \\
        \% chambers & 2.45892\% & 7.96986\% & 23.06303\% & 38.71193\% & 26.56897\% &
    \end{tabular}
    \caption{The number of samples for each $A_\lambda$ rank when $G = \ZZ/6\ZZ$.}
    \label{tab:rank-samples-6}
\end{table}

\subsection{Locally maximal chambers}\label{subsec:maximalChambers}
{ In our investigations, we were curious whether a greedy algorithm could be used to find an appropriate $\lambda$ by moving from chamber to chamber using a greedy algorithm.  For such a method to work, there should be no locally maximal chambers that do not achieve the maximum rank. A locally maximal chamber in terms of rank means that all adjacent chambers have rank less than or equal to the rank of the chamber.  Here, we describe our search for locally maximal chambers that do not achieve maximum rank.} To do this computationally, we cycle through every chamber in the hyperplane arrangement found by random sampling (as in \Cref{subsec:sampling}) and check whether the rank of this chamber is less than the global maximum, but greater than the rank of all adjacent chambers. The code is available in the file \texttt{locallyMaximalChambers.jl}.

The number of such chambers for small groups is given in \Cref{tab:locally-maximal}. For groups with $|G| \leq 5$, the rank of the locally maximal chambers { of deficient rank} is always one less than the dimension of the model. For $G = \mathbb{Z}/6\mathbb{Z}$, this appears to not be the case. For this group, our code did not complete, but had still found $625$ locally maximal chambers of ranks 23, 24, and 25 before it was terminated. However, since we were not able to verify that we had found all chambers in the hyperplane arrangement through random sampling, it may be the case that some of the locally maximal chambers found are not maximal at all. 

\begin{table}[!h]
    \centering
    \begin{tabular}{c|cccccc}
        Group & $\mathbb{Z}/3\mathbb{Z}$ & $\mathbb{Z}/4\mathbb{Z}$ & $\mathbb{Z}/2\mathbb{Z}\times \mathbb{Z}/2\mathbb{Z}$ & $\mathbb{Z}/5\mathbb{Z}$ & $\mathbb{Z}/6\mathbb{Z}$\\
        \hline
        Locally maximal chambers & 0 & 128& 0 & 1840 & $\geq 625$ \\
        \% & 0\% & 5.50\% & 0\% & 1.62\% & $\geq 0.006\%$  
    \end{tabular}
    \caption{The number of locally maximal chambers { of deficient rank} for each group. Locally maximal chambers { of deficient rank} are chambers which have non-maximal rank, and for which all adjacent chambers have rank lesser than or equal to the rank of the chamber.}
    \label{tab:locally-maximal}
\end{table}

\subsection{Weight vector for groups of even order}\label{subsec:weightVectorEven}

In \Cref{sec:solution}, we construct a weight vector $\lambda$ where $A_{\lambda}$ is the maximum rank possible for groups of odd order. Here, for even order groups, we construct an analogous vector $\lambda$ by following the construction in \Cref{sec:solution}, but including all elements of order $2$ in the set $X$. In this case, we find that the rank of $A_{\lambda}$ does not always achieve the empirically maximum rank. That is, random sampling sometimes finds $\lambda^\prime$ with $A_{\lambda^\prime}$ having greater rank than $A_{\lambda}$.

\begin{table}[h!]
    \centering
    \begin{tabular}{|r|c|c|c|}
        \hline
        Group & achieved rank & empirical max rank & gap \\
        \hline
        $\ZZ/4\ZZ$ & 15 & 16 & 1 \\%
        $\ZZ/6\ZZ$ & 26 & 26 & \\%
        $\ZZ/8\ZZ$ & 36 & 36 & \\%
        $\ZZ/10\ZZ$ & 46 & 46 & \\%
        $\ZZ/12\ZZ$ & 56 & 56 & \\%
        $\ZZ/14\ZZ$ & 66 & 66 & \\%
        $\ZZ/16\ZZ$ & 76 & 76 & \\%
        $\ZZ/18\ZZ$ & 86 & 86 & \\%
        $\left(\ZZ/2\ZZ\right)^2$ & 13 & 15 & 2 \\%
        $\left(\ZZ/2\ZZ\right)^3$ & 29 & 36 & 7 \\%
        $\left(\ZZ/2\ZZ\right)^4$ & 61 & 76 & 15 \\%
        $\ZZ/4\ZZ \times \left(\ZZ/2\ZZ\right)^2$ & 76 & 76 & \\%
        $\left(\ZZ/2\ZZ\right)^5$ & 125 & 156 & 31 \\%
        \hline
    \end{tabular}
    \caption{The rank of $A_{\lambda}$ for the ${\lambda}$ constructed in Section \ref{sec:solution} (achieved rank), the maximum rank of $A_\lambda$ found by randomly sampling $\lambda$ (empirical max rank), and the gap between the two ranks.}
    \label{tab:guess rank even}
\end{table}

Observe that in all cases, through random sampling we are able to find a weight vector where $A_\lambda$ has the maximum possible rank according to \Cref{lem:upper}. This means that in all cases in \Cref{tab:guess rank even}, the dimension of the variety is equal to the expected dimension of $5\ell + 1$.

We see that for  $\lambda$ constructed according to the methods in Section \ref{sec:solution}, $A_{\lambda}$ often achieves the maximum empirical rank, and it is only for powers of $\mathbb{Z}/2\mathbb{Z}$ that this construction does not work. Furthermore, in these cases, the difference between the rank of $A_{\lambda}$ and the maximum rank is equal to $|G|-1$. In \Cref{sec:example matrices}, we provide an example comparing the constructed $\lambda$ and the empirically maximal $\lambda$ for $G = \ZZ/4\ZZ$. Specifically, we find a weight vector $\lambda^\prime$ for which the corresponding matrix, $A_{\lambda^\prime}$, achieves the maximum rank of 16, and compare it to the matrix $A_\lambda$ obtained from a construction analogous to that in \Cref{sec:solution}. This gives good evidence that it is always possible to achieve the maximum possible rank (i.e., $5\ell + 1$) as the rank of $A_\lambda$, when $|G| \geq 5$.

\section{Discussion}\label{sec:discussion}

In this paper, we give a dimension formula for varieties associated to a 3-sunlet phylogenetic network and general group-based model, where the group $G$ is a finite abelian group of odd order. To do this, we use ideas from tropical geometry and linear algebra. Our proof relies on the fact that for odd order groups, there are no elements of order 2, and thus the non-identity elements can be partitioned into two sets of equal size (one of which we refer to as $X$ in \Cref{sec:solution}), each containing mutually inverse elements. Thus, our proof does not obviously generalise to even-order groups, and a full understanding of the dimension of these models remains open. In Section \ref{subsec:weightVectorEven} we construct weight vectors for even-order groups by assigning self-inverse elements to $X$ and following the construction in \Cref{sec:solution}. However, as shown in \Cref{tab:guess rank even}, for those groups that are products of $\mathbb{Z}/2\mathbb{Z}$, the rank of the corresponding $A_\lambda$ is less than the dimension of the model. Interestingly, the difference is $|G|-1$.

We have not yet explored models for which the subgroup $B\subset\Aut{G}$ is non-trivial. In \cite{gross2023dimensions}, a dimension formula is given for triangle-free phylogenetic networks for all group-based models (i.e., for all such subgroups $B$). This is achieved by choosing a weight vector that, for each edge in the network, is constant on the coordinates associated to the $B$-orbits on that edge. Since the structure of $\Aut(G)$ varies with $G$, this is only possible if we consider the two orbits $\{0\}$ and $G\setminus\{0\}$. However, in our experiments we found that for the 3-sunlet, the weight vectors $\lambda$ with $A_\lambda$ of maximal rank were not constant on these orbits, suggesting that a case-by-case analysis may be required.

The investigations in this paper highlight the intricate challenges involved in understanding 3-cycles. In the sampling experiments in Section \ref{sec:experiments}, we were able to examine the ranks of all chambers in the hyperplane arrangement. However, it is only when $|G| \geq 5$ that the expected dimension of the variety ($5\ell + 1$) is less than the dimension of the ambient space ($|G|^2$), so the cases with $|G|\leq 4$ are exceptional. We observe that as the group gets larger, there are proportionally more chambers of maximal rank. However, due to the large growth in the number of chambers as the size of the group grows, we were not able to determine if this is a general pattern.  

A first approach at obtaining the result for all finite abelian groups may be to try to adapt the weight vector $\lambda$.  Guidance on appropriate adaptations could be found by understanding how changes in the vector correspond to moving between chambers.  Indeed, a good understanding of the hyperplane arrangement may make it possible to devise an algorithm to search for a weight vector $\lambda$ for which $A_\lambda$ has maximal rank. However, as we have shown in \Cref{{subsec:maximalChambers}}, this is also not straightforward. In this section, we found that for at least some groups there are locally maximal chambers, and therefore a greedy algorithm starting at a lowest rank chamber and moving to chambers of strictly larger rank may not terminate on a globally maximal chamber.

\section{Acknowledgements}
SC was supported by the National Science Foundation Graduate Research Fellowship under Grant No. DGE-1841052, and by the National Science Foundation under Grant No. 1855135 during the writing of this paper.

SM was supported by the Biotechnology and Biological Sciences Research Council (BBSRC), part of UK Research and Innovation, through the Core Capability Grant BB/CCG1720/1 at the Earlham Institute and is grateful for HPC support from NBI’s Research Computing group. SM is grateful for further funding from BBSRC (grant number BB/X005186/1) which also supported this work.

EG was supported by the National Science Foundation grant DMS-1945584. 

This project was initiated at the ``Algebra of Phylogenetic Networks Workshop" held the University of Hawai`i at M\={a}noa and supported by National Science Foundation grant DMS-1945584. Additional parts of this research was performed while EG and SC were visiting the Institute for Mathematical and Statistical Innovation (IMSI) for the semester-long program on ``Algebraic Statistics and Our Changing World," IMSI is supported by the National Science Foundation (Grant No. DMS-1929348).
\bibliographystyle{plain}
\bibliography{bibliography.bib}

\begin{thebibliography}{10}

\bibitem{barton2022statistical}
Travis Barton, Elizabeth Gross, Colby Long, and Joseph Rusinko.
\newblock Statistical learning with phylogenetic network invariants.
\newblock {\em arXiv preprint arXiv:2211.11919}, 2022.

\bibitem{bezanson2017julia}
Jeff Bezanson, Alan Edelman, Stefan Karpinski, and Viral~B Shah.
\newblock Julia: A fresh approach to numerical computing.
\newblock {\em SIAM review}, 59(1):65--98, 2017.

\bibitem{cummings2023pfaffian}
Joseph Cummings, Elizabeth Gross, Benjamin Hollering, Samuel Martin, and Ikenna
  Nometa.
\newblock The {P}faffian structure of {CFN} phylogenetic networks.
\newblock {\em arXiv preprint arXiv:2312.07450}, 2023.

\bibitem{cummings2021invariants}
Joseph Cummings, Benjamin Hollering, and Christopher Manon.
\newblock Invariants for level-1 phylogenetic networks under the
  {C}avendar-{F}arris-{N}eyman model.
\newblock {\em Advances in Applied Mathematics}, 153:102633, 2024.

\bibitem{DRAISMA2008349}
Jan Draisma.
\newblock A tropical approach to secant dimensions.
\newblock {\em Journal of Pure and Applied Algebra}, 212(2):349--363, 2008.

\bibitem{evans1993invariants}
Steven~N Evans and T~P Speed.
\newblock Invariants of some probability models used in phylogenetic inference.
\newblock {\em The Annals of Statistics}, 21(1):355--377, 1993.

\bibitem{frohn2024invariants}
Martin Frohn, Niels Holtgrefe, Leo van Iersel, Mark Jones, and Steven Kelk.
\newblock Invariants for level-1 phylogenetic networks under the random walk
  4-state markov model.
\newblock {\em arXiv preprint arXiv:2407.11720}, 2024.

\bibitem{gross2023dimensions}
Elizabeth Gross, Robert Krone, and Samuel Martin.
\newblock Dimensions of level-1 group-based phylogenetic networks.
\newblock {\em arXiv preprint arXiv:2307.15166}, 2023.

\bibitem{gross2018distinguishing}
Elizabeth Gross and Colby Long.
\newblock Distinguishing phylogenetic networks.
\newblock {\em SIAM J. Appl. Algebra Geom.}, 2(1):72--93, 2018.

\bibitem{gross2021distinguishing}
Elizabeth Gross, Leo van Iersel, Remie Janssen, Mark Jones, Colby Long, and
  Yukihiro Murakami.
\newblock Distinguishing level-1 phylogenetic networks on the basis of data
  generated by {M}arkov processes.
\newblock {\em Journal of Mathematical Biology}, 83(32), 2021.

\bibitem{hollering2021identifiability}
Benjamin Hollering and Seth Sullivant.
\newblock Identifiability in phylogenetics using algebraic matroids.
\newblock {\em Journal of Symbolic Computation}, 104:142--158, 2021.

\bibitem{martin2023algebraic}
Samuel Martin, Vincent Moulton, and Richard~M Leggett.
\newblock Algebraic invariants for inferring 4-leaf semi-directed phylogenetic
  networks.
\newblock {\em bioRxiv:2023.09.11.557152}, 2023.

\bibitem{Nakhleh2011}
Luay Nakhleh.
\newblock {\em Problem Solving Handbook in Computational Biology and
  Bioinformatics}, chapter Evolutionary Phylogenetic Networks: Models and
  Issues, pages 125--158.
\newblock Springer Science+Business Media, LLC, 2011.

\bibitem{OSCAR}
Oscar -- open source computer algebra research system, version 1.0.0, 2024.

\bibitem{sturmfels1996grobner}
B~Sturmfels.
\newblock {\em Gr{\"o}bner Bases and Convex Polytopes}, volume~8 of {\em
  Universty Lectures Series}.
\newblock American Mathematical Society, Providence, RI, 1996.

\bibitem{sturmfels2005toric}
Bernd Sturmfels and Seth Sullivant.
\newblock Toric ideals of phylogenetic invariants.
\newblock {\em Journal of Computational Biology}, 12(4):457--481, 2005.

\bibitem{gsm194sullivant}
Seth Sullivant.
\newblock {\em Algebraic Statistics}, volume 194 of {\em Graduate Studies in
  Mathematics}.
\newblock American Mathematical Society, Providence, RI, 2018.

\bibitem{szekely1993fourier}
L~A Sz{\'e}kely, M~A Steel, and P~L Erd{\H o}s.
\newblock {F}ourier calculus on evolutionary trees.
\newblock {\em Advances in Applied Mathematics}, 14:200--216, 1993.

\end{thebibliography}
\newpage

\appendix

\section{Example Weights and Matrices}
\label{sec:example matrices}

\subsection{Example when $G = \ZZ/4\ZZ$.}

Below, we study the $\lambda$-construction from Section 3, adapted for odd even order groups (so all order two elements are in $X$), for $G = \ZZ/4\ZZ$. We denote this weight vector by $\lambda_{\text{guess}}$ and denote the corresponding matrix by $A_{\text{guess}}$. As noted in \Cref{tab:guess rank even}, $\mbox{rk} A_{\text{guess}} = 15$, which is not the empirical maximum. Through random sampling, we find in \Cref{tab:rank-samples-4} that there are 112 chambers of the $\ZZ/4\ZZ$-sunlet arrangement whose corresponding matrices achieve the maximal rank. We pick a weight vector, $\lambda_{\text{max}}$, in one of these chambers so that the corresponding matrix, $A_{\text{max}}$, has maximal rank and differs from $A_{\text{guess}}$ in exactly one column. For the $\lambda_{\text{max}}$ chosen here, the two matrices differ only in the column $[[2], [1], [1]]$.

\[A_{\text{max}} = \begin{array}{*{17}c}
    & \mathcal{T}_2 & \mathcal{T}_1 & \mathcal{T}_1 & \mathcal{T}_1 & \mathcal{T}_2 & \mathcal{T}_1 & \mathcal{T}_2 & \mathcal{T}_2 & \mathcal{T}_2 & \mathcal{T}_2 & \mathcal{T}_2 & \mathcal{T}_1 & \mathcal{T}_2 & \mathcal{T}_2 & \mathcal{T}_1 & \mathcal{T}_2 \\
	& 000 & 310 & 220 & 130 & 301 & 211 & 121 & 031 & 202 & 112 & 022 & 332 & 103 & 013 & 323 & 233 \\
	w_2^0 & 1 & 0 & 0 & 0 & 1 & 0 & 0 & 0 & 1 & 0 & 0 & 0 & 1 & 0 & 0 & 0 \\
	w_2^1 & 0 & 1 & 0 & 0 & 0 & 1 & 0 & 0 & 0 & 1 & 0 & 0 & 0 & 1 & 0 & 0 \\
	w_2^2 & 0 & 0 & 1 & 0 & 0 & 0 & 1 & 0 & 0 & 0 & 1 & 0 & 0 & 0 & 1 & 0 \\
	w_2^3 & 0 & 0 & 0 & 1 & 0 & 0 & 0 & 1 & 0 & 0 & 0 & 1 & 0 & 0 & 0 & 1 \\
	w_3^0 & 1 & 1 & 1 & 1 & 0 & 0 & 0 & 0 & 0 & 0 & 0 & 0 & 0 & 0 & 0 & 0 \\
	w_3^1 & 0 & 0 & 0 & 0 & 1 & 1 & 1 & 1 & 0 & 0 & 0 & 0 & 0 & 0 & 0 & 0 \\
	w_3^2 & 0 & 0 & 0 & 0 & 0 & 0 & 0 & 0 & 1 & 1 & 1 & 1 & 0 & 0 & 0 & 0 \\
	w_3^3 & 0 & 0 & 0 & 0 & 0 & 0 & 0 & 0 & 0 & 0 & 0 & 0 & 1 & 1 & 1 & 1 \\
	w_4^0 & 1 & 1 & 1 & 1 & 1 & 0 & 0 & 0 & 1 & 0 & 0 & 0 & 1 & 0 & 0 & 0 \\
	w_4^1 & 0 & 0 & 0 & 0 & 0 & 0 & 0 & 0 & 0 & 1 & 0 & 0 & 0 & 1 & 1 & 0 \\
	w_4^2 & 0 & 0 & 0 & 0 & 0 & 0 & 1 & 0 & 0 & 0 & 1 & 1 & 0 & 0 & 0 & 0 \\
	w_4^3 & 0 & 0 & 0 & 0 & 0 & 1 & 0 & 1 & 0 & 0 & 0 & 0 & 0 & 0 & 0 & 1 \\
	w_5^0 & 0 & 0 & 0 & 0 & 0 & 0 & 0 & 0 & 0 & 0 & 0 & 0 & 0 & 0 & 0 & 0 \\
	w_5^1 & 0 & 0 & 0 & 1 & 0 & 0 & 0 & 0 & 0 & 0 & 0 & 0 & 0 & 0 & 0 & 0 \\
	w_5^2 & 0 & 0 & 1 & 0 & 0 & 1 & 0 & 0 & 0 & 0 & 0 & 0 & 0 & 0 & 0 & 0 \\
	w_5^3 & 0 & 1 & 0 & 0 & 0 & 0 & 0 & 0 & 0 & 0 & 0 & 1 & 0 & 0 & 1 & 0 \\
	w_6^0 & 1 & 0 & 0 & 0 & 0 & 0 & 0 & 1 & 0 & 0 & 1 & 0 & 0 & 1 & 0 & 0 \\
	w_6^1 & 0 & 0 & 0 & 0 & 0 & 0 & 1 & 0 & 0 & 1 & 0 & 0 & 1 & 0 & 0 & 0 \\
	w_6^2 & 0 & 0 & 0 & 0 & 0 & 0 & 0 & 0 & 1 & 0 & 0 & 0 & 0 & 0 & 0 & 1 \\
	w_6^3 & 0 & 0 & 0 & 0 & 1 & 0 & 0 & 0 & 0 & 0 & 0 & 0 & 0 & 0 & 0 & 0 \\
\end{array}\]

\begin{align*}
    \left( \lambda_4^0, \lambda_4^1, \lambda_4^2, \lambda_4^3 \right) &= (-10, 0.773491, 0.278228, 0.313126) \\
    \left( \lambda_5^0, \lambda_5^1, \lambda_5^2, \lambda_5^3 \right) &= (0, 0, 0, 0) \\
    \left( \lambda_6^0, \lambda_6^1, \lambda_6^2, \lambda_6^3 \right) &= (-1, -5, -5, 5)
\end{align*}

\[A_{\text{guess}} = \begin{array}{*{17}c}
	& \mathcal{T}_2 & \mathcal{T}_1 & \mathcal{T}_1 & \mathcal{T}_1 & \mathcal{T}_2 & \mathcal{T}_2 & \mathcal{T}_2 & \mathcal{T}_2 & \mathcal{T}_2 & \mathcal{T}_2 & \mathcal{T}_2 & \mathcal{T}_1 & \mathcal{T}_2 & \mathcal{T}_2 & \mathcal{T}_1 & \mathcal{T}_2 \\
	& 000 & 310 & 220 & 130 & 301 & 211 & 121 & 031 & 202 & 112 & 022 & 332 & 103 & 013 & 323 & 233 \\
	w_2^0 & 1 & 0 & 0 & 0 & 1 & 0 & 0 & 0 & 1 & 0 & 0 & 0 & 1 & 0 & 0 & 0 \\
	w_2^1 & 0 & 1 & 0 & 0 & 0 & 1 & 0 & 0 & 0 & 1 & 0 & 0 & 0 & 1 & 0 & 0 \\
	w_2^2 & 0 & 0 & 1 & 0 & 0 & 0 & 1 & 0 & 0 & 0 & 1 & 0 & 0 & 0 & 1 & 0 \\
	w_2^3 & 0 & 0 & 0 & 1 & 0 & 0 & 0 & 1 & 0 & 0 & 0 & 1 & 0 & 0 & 0 & 1 \\
	w_3^0 & 1 & 1 & 1 & 1 & 0 & 0 & 0 & 0 & 0 & 0 & 0 & 0 & 0 & 0 & 0 & 0 \\
	w_3^1 & 0 & 0 & 0 & 0 & 1 & 1 & 1 & 1 & 0 & 0 & 0 & 0 & 0 & 0 & 0 & 0 \\
	w_3^2 & 0 & 0 & 0 & 0 & 0 & 0 & 0 & 0 & 1 & 1 & 1 & 1 & 0 & 0 & 0 & 0 \\
	w_3^3 & 0 & 0 & 0 & 0 & 0 & 0 & 0 & 0 & 0 & 0 & 0 & 0 & 1 & 1 & 1 & 1 \\
	w_4^0 & 1 & 1 & 1 & 1 & 1 & 0 & 0 & 0 & 1 & 0 & 0 & 0 & 1 & 0 & 0 & 0 \\
	w_4^1 & 0 & 0 & 0 & 0 & 0 & 1 & 0 & 0 & 0 & 1 & 0 & 0 & 0 & 1 & 1 & 0 \\
	w_4^2 & 0 & 0 & 0 & 0 & 0 & 0 & 1 & 0 & 0 & 0 & 1 & 1 & 0 & 0 & 0 & 0 \\
	w_4^3 & 0 & 0 & 0 & 0 & 0 & 0 & 0 & 1 & 0 & 0 & 0 & 0 & 0 & 0 & 0 & 1 \\
	w_5^0 & 0 & 0 & 0 & 0 & 0 & 0 & 0 & 0 & 0 & 0 & 0 & 0 & 0 & 0 & 0 & 0 \\
	w_5^1 & 0 & 0 & 0 & 1 & 0 & 0 & 0 & 0 & 0 & 0 & 0 & 0 & 0 & 0 & 0 & 0 \\
	w_5^2 & 0 & 0 & 1 & 0 & 0 & 0 & 0 & 0 & 0 & 0 & 0 & 0 & 0 & 0 & 0 & 0 \\
	w_5^3 & 0 & 1 & 0 & 0 & 0 & 0 & 0 & 0 & 0 & 0 & 0 & 1 & 0 & 0 & 1 & 0 \\
	w_6^0 & 1 & 0 & 0 & 0 & 0 & 0 & 0 & 1 & 0 & 0 & 1 & 0 & 0 & 1 & 0 & 0 \\
	w_6^1 & 0 & 0 & 0 & 0 & 0 & 0 & 1 & 0 & 0 & 1 & 0 & 0 & 1 & 0 & 0 & 0 \\
	w_6^2 & 0 & 0 & 0 & 0 & 0 & 1 & 0 & 0 & 1 & 0 & 0 & 0 & 0 & 0 & 0 & 1 \\
	w_6^3 & 0 & 0 & 0 & 0 & 1 & 0 & 0 & 0 & 0 & 0 & 0 & 0 & 0 & 0 & 0 & 0 \\
\end{array}\]

\end{document}